\documentclass[11pt]{article}
\usepackage[margin=1.25in]{geometry}
\RequirePackage{amsthm,amsmath,amsfonts,amssymb}
\RequirePackage{xcolor}
\RequirePackage{bbm}
\RequirePackage{gensymb}
\RequirePackage{centernot}
\RequirePackage{bm}
\RequirePackage{enumerate}
\RequirePackage{enumitem}
\RequirePackage{graphicx}
\RequirePackage{multirow}
\RequirePackage{booktabs}
\RequirePackage{mathtools}
\RequirePackage[authoryear]{natbib}

\theoremstyle{plain}
\newtheorem{theorem}{Theorem}
\newtheorem{corollary}{Corollary}
\newtheorem{claim}{Claim}
\newtheorem{lemma}{Lemma}
\newtheorem{proposition}{Proposition}

\theoremstyle{remark}

\newtheorem{remark}{Remark}
\newtheorem{assumption}{Assumption}

\newcommand{\Var}{\textnormal{Var}}

\newcommand{\eff}{\textnormal{Eff}}

\newcommand{\opt}{\textnormal{opt}}

\newcommand{\sign}{\textnormal{sign}}

\newcommand{\ovl}{\overline}
\newcommand{\ol}{\widetilde}
\newcommand{\indic}{\bm{1}}
\newcommand{\ind}{\bm{1}}

\newcommand{\ca}{\mathcal{A}}
\newcommand{\cf}{\mathcal{F}}
\newcommand{\ci}{\mathcal{I}}
\newcommand{\cj}{\mathcal{J}}
\newcommand{\ct}{\mathcal{T}}

\renewcommand{\le}{\leqslant}
\renewcommand{\leq}{\leqslant}
\renewcommand{\ge}{\geqslant}
\renewcommand{\geq}{\geqslant}
\newcommand{\dunif}{\mathbb{U}}

\newcommand{\Rbar}{\overline{\mathbb{R}}}

\newcommand{\e}{\mathbb{E}}
\newcommand{\real}{\mathbb{R}}
\newcommand{\rd}{\,\mathrm{d}}

\newcommand{\giv}{\!\mid\!}
\newcommand{\err}{\varepsilon}
\newcommand{\tran}{\mathsf{T}}
\newcommand{\intr}{\textnormal{int}}

\newcommand{\cx}{\mathcal{X}}

\newcommand{\phe}{\phantom{=}}
\newcommand{\phz}{\phantom{0}}
\newcommand{\toas}{\stackrel{\mathrm{a.s.}}{\to}}
\newcommand{\var}{\mathrm{Var}}

\newcommand{\bsp}{\boldsymbol{p}}
\DeclareMathOperator*{\argmin}{arg\,min}
\DeclareMathOperator*{\argmax}{arg\,max}

\title{A general characterization of optimal tie-breaker designs}
\author{Harrison H. Li\\Stanford University
 \and
Art B. Owen\\
Stanford University
}
\date{October 2022}
\begin{document}
\maketitle
\begin{abstract}
Tie-breaker designs trade off a statistical design objective with
short-term gain from preferentially assigning a binary treatment to those with high values of a running variable $x$.
The design objective is any continuous function of the expected information matrix in a two-line regression model,
and short-term gain is expressed as the covariance between the running variable and the treatment indicator.
We investigate how to specify design functions indicating treatment probabilities as a function of $x$ to optimize these competing objectives,
under external constraints on the number of subjects receiving treatment.
Our results include sharp existence and uniqueness guarantees,
while accommodating the ethically appealing
requirement that treatment probabilities are non-decreasing in $x$.
Under such a constraint,
there always exists an optimal design function
that is constant below and above a single discontinuity.
When the running variable distribution is not symmetric or the fraction of subjects receiving the treatment is not $1/2$,
our optimal designs improve upon a $D$-optimality objective
without sacrificing short-term gain,
compared to the three level tie-breaker designs of Owen and Varian (2020)
that fix treatment probabilities at $0$, $1/2$, and $1$.
We illustrate our optimal designs
with data from Head Start, 
an early childhood government intervention program.

\end{abstract}




\section{Introduction}
\label{sec:introduction}
Companies, charitable institutions, and clinicians often have ethical or economic reasons to prefer assigning a binary treatment to certain individuals.
If this preference is expressed by the values of a scalar
running variable $x$,
a natural decision is to assign the treatment to a subject 
if and only if their $x$ is at least some threshold $t$.
This is a regression discontinuity
design, or RDD~\citep{this:camp:1960}.
Unfortunately,
treatment effect estimates from an RDD analysis typically have
very high variance \citep{jacob2012practical,gold:1972,gelman2017high},
relative to those from a randomized control trial (RCT)
that does not preferentially treat any individuals.
To trade off between these competing statistical and ethical objectives,
investigators can use a tie-breaker design (TBD).
In a typical tie-breaker design, the top ranked subjects get the
treatment, the lowest ranked subjects are in the control group
and a cohort in the middle are randomized to treatment or control.
The earliest tie-breaker reference that we are aware of
is \cite{camp:1969}
where $x$ was discrete and the randomization broke
ties among subjects with identical values of $x$.

Past settings for tie-breaker
designs include the offer of remedial English
to incoming university students based on their
high school English proficiency
\citep{aike:west:schw:carr:hsiu:1998},
a diversion program designed to reduce juvenile delinquency~\citep{lips:cord:berg:1981},
scholarship offerings for two
and four year colleges 
based on a judgment of the applicants' needs and academic strengths~\citep{abdulkadiroglu2017impact, angrist2020},
and clinical trials~\citep{Trochim92}, where they are known
as cutoff designs.

The tie-breaker design problem
is to choose
treatment probabilities $p_i$ for subjects $i=1,\dots,n$ based on their
running variables $x_i$.
These probabilities are chosen before observing the response values $y_1,\dots,y_n$
but with the running variables $x_1,\dots,x_n$ known.
We assume throughout that
$p_i=p_{i'}$ whenever $x_i=x_{i'}$.

As is common in the optimal experimental design literature,
the statistical objective is an ``efficiency" criterion that measures estimation precision.
Specifically, our criterion 
will be a function
$\Psi(\cdot)$ of the information 
(scaled inverse variance)
matrix $\ci_n(p_1,\dots,p_n)$ for the model parameters $\beta = (\beta_0,\beta_1,\beta_2,\beta_3)^{\top}$ in
a two line model relating the response $y_i$
to the running variable $x_i$ and a treatment indicator
$z_i\in\{-1,1\}$:
\begin{equation}
\label{eq:two_line_model}
y_i = \beta_0 + \beta_1 x_i + \beta_2 z_i + \beta_3 x_i z_i + \err_i.
\end{equation}
This simple working model nonetheless poses some
challenging design problems.  
In Section~\ref{sec:discussion}
we describe some more general modeling settings for tie-breaker models.

Throughout we assume that the running variable is centered, i.e.\ $(1/n)\sum_ix_i=0$,
and that the $\err_i$ have common variance $\sigma^2$.
Here $z_i=1$ indicates treatment and so $p_i = \Pr(z_i=1) = (1+\e(z_i))/2$).
For model~\eqref{eq:two_line_model}
the information matrix is $\ci_n = \e(\cx_i\cx_i^{\top})$ 
where $\cx_i = (1,x_i,z_i,x_iz_i)^{\top} \in \mathbb{R}^4$
and the expectation is taken over the treatment assignments $z_i$,
conditional on the running variables $x_i$ (whose values are known).
The ordinary least squares
estimate $\hat{\beta}$ of $\beta$
satisfies $\e(\var(\hat\beta)^{-1})=n\ci_n/\sigma^2$.
Common examples of efficiency criteria $\Psi(\cdot)$ in the literature,
such as the D-optimality criterion $\Psi_D(\cdot) = \log(\det(\cdot))$,
are concave in both $\ci_n$ and $p = (p_1,\ldots,p_n)^{\top}$~\citep{boyd:vand:2004}.
However, our theoretical results only require continuity of $\Psi(\cdot)$.

Our preference for treating individuals with
higher running variables $x$ is expressed
as an equality constraint on
the scaled covariance $\ovl{xp}\equiv(1/n)\sum_{i=1}^n x_ip_i$
between treatment and the running variable
(recall the latter is known, hence viewed as non-random).
Under the two-line model~\eqref{eq:two_line_model}, this
constraint has the following economic interpretation.
We take $y$ to be something like
economic value or student success, where larger $y$
is better.
We expect that $\beta_3>0$ holds in most of our motivating problems.
The expected value of $y$ per customer under~\eqref{eq:two_line_model} is then
\begin{align}
\label{eq:short_term_gain}
\e(y_i) = \beta_0 + \beta_2 \cdot (2\bar{p}-1) + \beta_3 \cdot (2\ovl{xp}-1)
\end{align}
where $\bar{p} \equiv (1/n)\sum_{i=1}^n p_i$.
Equation~\eqref{eq:short_term_gain} shows that the expected gain is unaffected by $\beta_0$ or $\beta_1$.
Furthermore,
we assume the proportion of treated subjects is fixed by an external budget, 
i.e., an equality constraint $\bar{p} = \ol{p}$ for some $\ol{p} \in (0,1)$.
For instance, there might be only a set number of scholarships or perks to be given out.
The only term affected by the design in~\eqref{eq:short_term_gain} is
then $\beta_3 \cdot \ovl{xp}$, as pointed out by~\citet{owen:vari:2020}.
For $\beta_3>0$, the short term average value
per customer grows with $\ovl{xp}$ and we would
want that value to be large.
Similar functionals are also commonly studied as regret functions in bandit problems~\citep{goldenshlugerzeevi2013, metelkinapronzato2017}.

We are now ready to formulate the tie-breaker design problem as
the following constrained optimization problem.
Given real values $x_1\le x_2\le \cdots\le x_n$:
\begin{equation}
\label{eq:p_opt_problem_fd}
\begin{array}{ll}
\mbox{maximize}
&\qquad \Psi(\ci_n(\bsp))  \\
\mbox{over}&\qquad \bsp=(p_1,\dots,p_n)\in\ca\\
\mbox{subject to} &\qquad n^{-1} \sum_{i=1}^n p_i = \ol{p} \qquad \mbox{and} \qquad n^{-1} \sum_{i=1}^n x_ip_i = \ol{xp}
\end{array}
\end{equation}
for some constants $\ol{p}$ and $\ol{xp}$.
The first equality constraint in~\eqref{eq:p_opt_problem_fd} is a budget constraint
due to the cost of treatment and the second constraint
is on the short term gain mentioned above.
We consider two different sets $\ca$ in detail.
The first is $[0,1]^n$.  
The second is
$\{ \bsp\in[0,1]^n\mid 0\le p_1\le p_2\le \cdots \le p_n\le 1\}$
which requires treatment probabilities to be non-decreasing in the running variable $x$.
Such a monotonicity constraint
prevents more qualified students from having a lower chance of getting a scholarship than less qualified ones
or more loyal customers having a lower chance for a perk than others.
It also eliminates perverse incentives for subjects to lower their $x_i$.
To our knowledge, such a monotonicity constraint has not been received much attention in the optimal design literature,
though it is enormously appealing in our motivating applications.



When the efficiency criterion $\psi(\cdot)$ is concave in $\bsp$,
then a solution to~\eqref{eq:p_opt_problem_fd} can be found numerically via convex optimization,
as \cite{mvtiebreaker} do for vector valued $x_i$,
and as~\citet{metelkinapronzato2017} mention for a similar problem.
However, our particular setting with univariate $x_i$ is tractable enough to provide a simple yet complete analytical characterization of the optimal $p_i$,
even if the efficiency criterion is not concave.
We show, under general conditions,
that we can always find optimal treatment probabilities that are piecewise constant in $x$,
with the number of pieces small and independent of $n$.

There is a well-developed literature for optimal experiment design in the presence of multiple objectives.
Early examples of a constrained optimization problem of the form~\eqref{eq:p_opt_problem_fd} were designed to account for several of the standard efficiency objectives simultaneously~\citep{stigler1971,lee1987,lee1988}.
\citet{lauter1974,lauter1976} proposed maximizing a convex combination of efficiency objectives,
a practice now typically referred to as a ``compound'' design approach.
It is now well known~\citep{cookwong1994,clydechaloner1996} that in many problems with concave objectives,
optimal constrained and compound designs are equivalent.
In this paper, we provide another approach to reduce the constrained problem~\eqref{eq:p_opt_problem_fd}
to a compound problem
that can handle the monotonicity constraint.
At the same time,
we provide simple ways to compute our optimal designs that are based directly on the parameters $\ol{p}$ and $\ol{xp}$ in our constrained formulation~\eqref{eq:p_opt_problem_fd},
and do not require specifying the
Lagrange multipliers appearing in the corresponding compound problem.
Those Lagrange multipliers involve ratios of
information gain to economic gain where each of those quantities
is only known up to a multiplicative constant.

Problems similar to~\eqref{eq:p_opt_problem_fd} have received significant attention in the sequential design of clinical trials.
Biased-coin designs, beginning with the simple procedure of~\citet{efron1971},
have been developed as a compromise between treatment balance and randomization; see~\citet{atkinson2014} for a review.
Covariate-adaptive biased-coin designs
often replace the balance objective
with an efficiency criterion such as D-optimality~\citep{atkinson1982, rosenbergersverdlov2008}.
Response-adaptive designs also optimize for some  efficiency objective
but simultaneously seek to minimize the number of patients receiving the inferior treatment for ethical reasons~\citep{hurosenberger2006}.
Various authors such as~\citet{bandyopadhyaybiswas2001} and~\citet{huetal2015} propose sequential designs to effectively navigate this trade-off.
When they also account for covariate information, they are called covariate-adjusted response-adaptive (CARA) designs~\citep{zhangetal2007, zhanghu2009}.

In the CARA literature especially, 
there has been significant recent interest in optimal design for nonlinear models~\citep{sverdlovetal2013,metelkinapronzato2017, biswasbhattacharya2018}.
Unlike optimal designs in linear models such as~\eqref{eq:two_line_model},
designs in nonlinear models can typically only be locally optimal,
meaning that their optimality depends on the values of the unknown parameters~\citep{chernoff1953}.
While we may be able to obtain increasingly reliable estimates of these parameters over time in sequential settings,
in non-clinical settings subjects typically enter a tie-breaker study non-sequentially,
i.e., we know the running variables for all subjects before designing the experiment.
In these applications --- such as measuring the impact of a scholarship on future educational attainment --- it can take several years to collect a single set of responses on which to compute a parameter estimate.
Locally optimal designs are therefore of limited utility in this setting,
and so we focus on optimal design under the linear model~\eqref{eq:two_line_model} in a non-sequential setting,
which already presents a sufficient challenge.

The existing literature on problems like~\eqref{eq:p_opt_problem_fd} typically considers the running variable $x$ to be random.
For example,
Section 7 of \cite{owen:vari:2020} study tie-breaker designs under the assumption that the running variable is either uniform or Gaussian,
and exactly half the subjects are to be treated. 
They consider the typical three level tie-breaker design where subjects with running variable $x$ above some threshold $\Delta$ always get the treatment,
subjects with running variable below $-\Delta$ never get the treatment,
and the remaining subjects are randomized into treatment with probability $1/2$.
They find that a $c$-optimality criterion of statistical
efficiency is monotonically increasing in the width
$\Delta$ of the randomization window,
with the RCT ($\Delta\to\infty$) being most efficient and the RDD ($\Delta=0$)
least efficient.
Conversely, the short-term gain is decreasing in $\Delta$.
They also show the three level design is optimal for any given level of short-term gain.
In this article we show strong 
advantages to moving away from that three level design
when the running variable is not symmetric,
or we cannot treat half of the subjects.

\citet{metelkinapronzato2017} studied a further generalization of the optimal tie-breaker design problem,
motivated by CARA designs.
In particular, their Example 1,
an illustration of their Corollary 2.1,
is similar\footnote{There are minor differences such as the lack of a treatment fraction constraint,
an inequality constraint on short-term gain as opposed to an equality constraint,
and the use of a common intercept for the treated and untreated individuals,
i.e., assuming that $\beta_2=0$ in~\eqref{eq:two_line_model}.} to a random-$x$ generalization of~\eqref{eq:p_opt_problem_fd}.
Crucially, however,
the proof for their Corollary 2.1 does not generalize to the case where we require the treatment probabilities be monotone.
Even without the monotonicity constraint,
we provide a sharper characterization of the solutions to our more specific problem
(Section~\ref{sec:global_optimal}).

We introduce a random-$x$ generalization of the problem~\eqref{eq:p_opt_problem_fd} in Section~\ref{sec:random_x}.
We show it encompasses both ~\eqref{eq:p_opt_problem_fd} 
and the problem studied by~\citet{owen:vari:2020}
as special cases.
Then,
Section~\ref{sec:characterizations}
presents the main technical results
characterizing the solutions to this more general problem.
In particular, Theorem~\ref{thm:alt_pmax_pmin_mon}
shows that,
under the monotonicity constraint,
there always exists a solution to~\eqref{eq:p_opt_problem_fd}
corresponding to a simple, two-level stratified design
where all subjects with running variable
below some threshold $t'$
have the same probability of treatment,
all subjects with running variable above $t'$
have an identical, higher probability of treatment,
and those subjects with $x=t'$ have a treatment probability between these two values.
Section~\ref{sec:inadmissibility} then presents some results on the trade-off between a $D$-optimality efficiency criterion
and short-term gain for these optimal designs;
examples of this trade-off
for some specific running variable distributions are given
in Section~\ref{sec:examples}.
That section also includes a fixed-$x$ application
based on Head Start,
a government assistance program for low-income children.
It also shows how to compute our optimal designs
when $x$ is either fixed or random.
Finally, Section~\ref{sec:discussion}
provides summarizes the main results.

\section{Random running variable}
\label{sec:random_x}
Our random-$x$ generalization of~\eqref{eq:p_opt_problem_fd}
assumes the running variables $x_i$
are samples from a common distribution $F$,
which we hereafter identify with the corresponding cumulative distribution function.
It considers an information matrix that averages over both the random treatment assignments and randomness in the running variables.
This allows us to characterize optimal designs for an as yet unobserved set of running variable values,
when their distribution is known.

Before presenting the random-$x$ tie-breaker design problem,
we briefly review the standard setting of optimal design in multiple linear regression models;
see e.g.~\citet{atkinsonetal2007} for further background.
In the simplest case, 
the user assumes the standard linear model $y_i = \bm{x}_i^{\tran} \bm{\beta} + \err_i$
with the goal of selecting covariate values $\bm{x}_1,\dots,\bm{x}_n \in \real^p$ 
to optimize an efficiency criterion that is a function of the information matrix $\cx^\tran \cx$, 
where $\cx \in \real^{n \times p}$ is the design matrix with $i$-th row $\bm{x}_i^{\tran}$. 
Perhaps the most common such criterion is D-optimality,
which corresponds to maximizing $\log(\det(\cx^\tran \cx))$.
Another popular choice is $c$-optimality, 
which minimizes $c^\tran(\cx^\tran\cx)^{-1}c$ for some choice of $c \in \real^p$. 
This can be interpreted as minimizing $\var(c^{\top}\hat{\beta} \mid \cx)$,
where $\hat{\beta}$ is the ordinary least squares estimator.
The study of optimal design is often simplified by the use of \textit{design measures}. 
A design measure $\xi$ is a probability distribution
from which to generate the covariates $\bm{x}_i$.
The relaxed optimal design problem involves selecting a design measure $\xi$ instead of a finite number of covariate values $\bm{x}_1,\ldots,\bm{x}_n$.
The objective is to optimize for the desired functional of the \textit{expected} information matrix $\ci(\xi) \equiv \e_{\xi}[\cx^{\tran}\cx]$
over some space $\Xi$ of design measures $\xi$. 
For instance, a design measure $\xi^*$ is D-optimal (for the relaxed problem) 
if $\xi^* \in \argmax_{\xi \in \Xi} \det(\ci(\xi))$, 
and $c$-optimal if $\xi^* \in \argmin_{\xi \in \Xi} c^{\top}\ci(\xi)^{-1}c$.
The original optimal design problem restricts $\Xi$ to only consist of discrete probability distributions supported on at most $n$ distinct points with probabilities that are multiples of $1/n$.


For the tie-breaker design problem,
our regression model~\eqref{eq:two_line_model}
includes both the running variables $x_i$
and the treatment indicators $z_i$ as covariates.
But the experimenter does not have control over the entire joint distribution of $(x_i,z_i)$.
The running variable is externally determined,
so they can only specify the conditional distribution of the treatment indicator $z_i$ given the running variable $x_i$.
This conditional distribution is specified by a \textit{design function} $p:\real \to [0,1]$
such that $p(x) \equiv \Pr(z_i = 1 \mid x_i=x)$.
As mentioned above, we assume $x_i \sim F$ for a known, fixed distribution $F$.
This allows us to drop subscripts $i$ when convenient.
For any two design functions $p$ and $p'$
we say $p=p'$ whenever $\Pr_F(\{x:p(x)=p'(x)\})=1$.
We only need a minimal assumption on $F$,
which can be continuous, discrete, or neither:
\begin{assumption}
\label{ass:F}
$0 < \var_F(x) < \infty$ with $\e_F(x)=0$.
\end{assumption}
The mean-centeredness part of Assumption~\ref{ass:F} loses no generality,
due to the translation invariance of estimation under the two-line model~\eqref{eq:two_line_model}.
All expectations involving $x$ hereafter omit the subscript $F$ from all such expectations with the implicit understanding that $x \sim F$.

The random-$x$ tie-breaker design problem is as follows:
\begin{equation}
\label{eq:p_opt_problem}
\begin{array}{ll}
\mbox{maximize} &\qquad \Psi(\ci(p))  \\
\mbox{over} &\qquad p\in \cf\\
\mbox{subject to} &\qquad \e_p(z) = \ol{z} \\
\mbox{and}&\qquad \e_p(xz) = \ol{xz}.
\end{array}
\end{equation}
Here $\ol{z}$ and $\ol{xz}$ are constants analogous to $\bar{p}$ and $\overline{xp}$, respectively, in~\eqref{eq:p_opt_problem_fd},
$\cf$ is a collection of design functions,
and $\ci(p)$ is the expected information matrix under the model~\eqref{eq:two_line_model},
averaging over both $x \sim F$ and $z \mid x \sim p$.

This problem can be viewed as a constrained relaxed optimal design problem under the regression model~\eqref{eq:two_line_model}
where the set $\Xi$ of allowable design measures is indexed by the design functions $p \in \cf$.

To interpret the equality constraints in~\eqref{eq:p_opt_problem} it is helpful to note that
\begin{align}
\e_p(x^az) =\e( x^a\e_p(z\giv x))=\e( x^a(2p(x)-1))=2\e(x^ap(x))-\e(x^a)\label{eq:tower}
\end{align}
for any $a \ge 0$ with $\e(|x|^a)<\infty$.
In particular, for each positive integer $a$,
there exists an invertible linear mapping $\varphi_a:\real^{a+1} \to \real^{a+1}$
that does not depend on the design $p$
and maps $(\e_p(z),\ldots,\e_p(x^a z))$
to $(\e(p(x)),\ldots,\e(x^ap(x)))$.
For example, $\varphi_1(x,y)=(1/2)(1+x,y)$.
Taking $a=0$ in~\eqref{eq:tower},
we see the constraint $\e_p(z)=\ol{z}$
in~\eqref{eq:p_opt_problem}
is equivalent to requiring the expected proportion of subjects to be treated to be $(1+\ol{z})/2$.
This proportion is typically determined by 
the aforementioned budget constraints.
Taking $a=1$ in~\eqref{eq:tower} 
shows that the second constraint $\e_p(xz)=\ol{xz}$ in~\eqref{eq:p_opt_problem}
sets the expected level of short-term gain.
In Section~\ref{subsection:stg_bounds},
we provide some guidance on how to choose $\ol{xz}$
in practice.

From computing the expected information matrix $\ci(p)$ in Section~\ref{sec:efficiency},
we will see the problem~\eqref{eq:p_opt_problem} reduces to the finite-dimensional problem~\eqref{eq:p_opt_problem_fd} 
when $F$ is discrete,
placing probability mass $n^{-1}$ on each of the known running variable values $x_1,\ldots,x_n$.
Thus, to solve~\eqref{eq:p_opt_problem_fd} it suffices to solve the problem~\eqref{eq:p_opt_problem} for any $F$ satisfying Assumption~\ref{ass:F},
which must hold for any discrete distribution with finite support.
\subsection{Some design functions}

For convenience, we introduce some notation for certain forms of the design function $p$.
We will commonly encounter designs of the form
\begin{align} \label{eq:p_A}
p_A(x) \equiv \ind(x\in A)
\end{align}
for a set $A \subseteq \real$.
Another important special case consists of two level designs
\begin{align}\label{eq:twolevel}
p_{\ell,u,t}(x) \equiv \ell \ind(x<t)+u\ind(x\ge t)
\end{align}
for treatment probabilities $0 \le \ell \le u \le 1$ and a threshold $t\in \Rbar$.
For example, $p_{0,1,t}$ is a sharp RDD with threshold $t$,
while for any $t$, 
$p_{\theta,\theta,t}$ is an RCT with treatment probability $\theta$.

The condition $\ell \le u$ ensures that $p(x)$ is nondecreasing in $x$; we refer to such designs as \textit{monotone}. 
Under a monotone design, a subject cannot have a lower treatment probability than another subject with lower $x$.
We also define a \textit{symmetric} design to be one for which
$p(-x)=1-p(x)$; for instance, $p$ might
be the cumulative distribution function (CDF)
of a symmetric random variable. Finally, the three level tie-breaker design from \cite{owen:vari:2020} is both monotone and symmetric and defined for $\Delta \in [0,1]$ by
\begin{equation}\label{eq:3_level_tie-breaker}
p_{3,\Delta}(x) \equiv 0.5 \indic(|x| \le \Delta) + \indic(x > \Delta)
\end{equation}
when $F$ is the $\dunif(-1,1)$ distribution.
Note that for all $\Delta$,
$p_{3,\Delta}$ always treats half the subjects,
i.e., $\ol{z}=0$.
The generalization to other $\ol{z}$ and running variable distribution functions $F$ is
\begin{equation}
\label{eq:3_level_quantile_tie-breaker}
p_{3;\ol{z},\Delta}(x) =
0.5\times\indic(a(\ol{z},\Delta) < x < b(\ol{z},\Delta))
+ \indic(x \ge b(\ol{z},\Delta))
\end{equation}
where $a(\ol{z},\Delta) = F^{-1}((1-\ol{z})/2-\Delta)$
and $b(\ol{z},\Delta)=F^{-1}((1-\ol{z})/2+\Delta)$.
\subsection{Bounds on short-term gain}
\label{subsection:stg_bounds}
Before studying optimal designs, we impose lower and upper
bounds on the possible short-term gain constraints $\ol{xz}$ to consider, for each possible $\ol{z} \in (-1, 1)$.
For an upper bound we use $\ol{xz}_{\max}(\ol{z})$, the maximum $\ol{xz}$ that can be attained by \textit{any} design function $p$ satisfying the treatment fraction constraint $\e_p(z)=\ol{z}$.
It turns out that this upper bound is always uniquely attained.
If the running variable distribution $F$ is continuous, it is uniquely attained by a sharp RDD.
We remind the reader that uniqueness of a design function satisfying some property means that for any two design functions $p$ and $p'$ with that property,
we must have $\Pr(p(x)=p'(x))=1$ under $x \sim F$.

\begin{lemma}
\label{lemma:stg}
For any $\ol{z} \in [-1, 1]$ and running variable distribution $F$, there exists a unique design $p_{\ol{z}}$ satisfying
\begin{align} 
\e_{p_{\ol{z}}}(z) & = \ol{z},\quad\text{and} \label{eq:budget_constraint} \\
p_{\ol{z}}(x) & =
\begin{cases}
1, & x > t \\
0, & x < t
\end{cases} \label{eq:generalized_rdd}
\end{align}
for some $t\in \real$. 
Any $p$ that satisfies the treatment fraction constraint~\eqref{eq:budget_constraint} also satisfies
\begin{align}\label{eq:stg}
\e_p(xz)
\leq \e_{p_{\ol{z}}}(xz) \equiv \ol{xz}_{\max}(\ol{z})
\end{align}
with equality if and only if $p=p_{\ol{z}}$,
i.e. $\Pr(p(x)=p_{\ol{z}}(x))=1$ under $x\sim F$.
\end{lemma}
\begin{remark}
Notice that equation~\eqref{eq:generalized_rdd} does not specify $p_{\ol{z}}(x)$ at $x=t$.
If $F$ is continuous, then
any value for $p_{\ol{z}}(t)$ yields an equivalent design function,
but if $F$ has an atom at $x=t$ then we will require a specific value for $p_{\ol{z}}(t)\in [0,1]$. 
We must allow $F$ to have atoms to solve the finite dimensional problem~\eqref{eq:p_opt_problem_fd}.
While we specify $p_{\ol{z}}(t)$ in the proof of Lemma~\ref{lemma:stg} below,
later results of this type do not
give the values of design functions at such discontinuities.
\end{remark}
\begin{remark}
\label{remark:generalized_RDD}
If $F$ is continuous, then $p_{\ol{z}}$ is an RDD: $p_{\ol{z}} = p_{0,1,t}$ for $t = F^{-1}((1-\ol{z})/2)$.
We call the design $p_{\ol{z}}$ a \textit{generalized RDD} for general $F$ satisfying Assumption~\ref{ass:F}.
\end{remark}
\begin{remark}
The threshold $t$ in~\eqref{eq:generalized_rdd} is essentially unique.  
If there is an interval $(t,s)$ with $\Pr(t<x<s)=0$ then all step locations in $[t,s)$ provide equivalent generalized RDDs.
\end{remark}

\begin{proof}[Proof of Lemma~\ref{lemma:stg}.]
If $\ol{z} \in \{-1,1\}$ then the only design functions (again, up to uniqueness w.p.1 under $x \sim F$) are the constant functions $p(x)=0$ and $p(x)=1$, and the result holds trivially.
Thus, we can assume that $\ol{z} \in (-1,1)$.
By~\eqref{eq:tower}, the existence of $p_{\ol{z}}$ follows by taking $t = \inf\{s:F(s) \geq (1-\ol{z})/2\}$ and 
\[
p_{\ol{z}}(t) =
\begin{cases}
0, & \text{ if $\Pr(x=t) = 0$} \\
\frac{F(t)-(1-\ol{z})/2}{\Pr(x=t)}, & \text{ if $\Pr(x=t) > 0$.}
\end{cases}
\]
To show~\eqref{eq:stg}, fix any design $p$ satisfying~\eqref{eq:budget_constraint} and notice that $\e(p(x)-p_{\ol{z}}(x))=0$ means
\[
\e(p(x)\ind(x < t))+(p(t)-p_{\ol{z}}(t))\Pr(x=t)+\e((p(x)-1)\ind(x>t))=0.
\]
Then $\e(x(p_{\ol{z}}(x)-p(x)))$ equals
\begin{align*}
 &\phantom{\ge}\  \e(-xp(x)\ind(x < t)) + t(p_{\ol{z}}(t)-p(t))\Pr(x=t) + \e(x(1-p(x))\ind(x > t))\\
& \ge t[\e(-p(x)\ind(x<t))+(p_{\ol{z}}(t)-p(t))\Pr(x=t)+\e((1-p(x))\ind(x > t))] \\
& = 0
\end{align*}
with equality iff $(t-x)p(x)\indic(x<t)=(x-t)(1-p(x))\indic(x>t)=0$ 
for a set of $x$ with probability one under $F$, i.e., iff $p$ satisfies~\eqref{eq:generalized_rdd}
with probability one under $x\sim F $.
\end{proof}
By symmetry,
the design that minimizes $\e_p(xz)$ over all designs $p$ with $\e_p(z)=\ol{z}$ is $p_{1,0,s}$ where $s = F^{-1}((1+\ol{z})/2)$. 
Notice that $\e_{p_{1,0,s}}(xz) = \ol{xz}_{\min}(\ol{z}) \equiv 2\e(x\indic(x < s)) < 0$.
We impose a stricter lower bound of $\ol{xz} \geq 0$ in the context of problem~\eqref{eq:p_opt_problem}.
This is motivated by the fact that the running variable $x$ has mean 0 (Assumption~\ref{ass:F}),
meaning that $\e_p(xz)=0$ whenever the design function $p$ is constant,
corresponding to an RCT.
Designs with $\ol{xz} < 0$ exist for all $\ol{z} \in (-1,1)$ but would not be relevant in our motivating applications, 
as they represent scenarios where subjects with \textit{smaller} $x$ are more preferentially treated than in an RCT.
We hence define the \textit{feasible input space} $\cj$ by
\begin{equation}\label{eq:defj}
\cj \equiv \Bigl\{(\ol{z}, \ol{xz})\mid
-1 < \ol{z} < 1,\
0 \leq \ol{xz} \leq \ol{xz}_{\max}(\ol{z})
\Bigr\} \subseteq \mathbb{R}^2.
\end{equation}
Any design function $p$ for which the moments $(\e_p(z),\e_p(xz))$
lie within the feasible input space $\cj$
is referred to as an \textit{input-feasible} design function.

If the design $p$ is input-feasible,
we can write $\e_p(xz) = \delta \cdot \ol{xz}_{\max}(\e_p(z))$
for some $\delta \in [0,1]$.
The parameter $\delta$ corresponds to the amount of additional short-term gain
attained by the design $p$ over an RCT,
relative to the amount of additional short-term gain attained by the generalized RDD $p_{\e_p(z)}$
that treats the same proportion of subjects as $p$.
For instance,
$\delta=0.4$ means that the design $p$ has a short-term gain that is 40\% of the way
from that of an RCT 
to the maximum attainable short-term gain
under the treatment fraction constraint.

\subsection{Expected information matrix and equivalence of $D$-optimality and $c$-optimality}
\label{sec:efficiency}
We now explicitly compute the expected information matrix
\begin{equation}
\label{eq:I}    
\ci(p) =
\sigma^{-2} \e(n^{-1} \cx^{\tran} \cx)
=
\sigma^{-2}
\begin{pmatrix}
1 & 0 & \e(z) & \e(xz) \\
0 & \e(x^2) & \e(xz) & \e(x^2z) \\
\e(z) & \e(xz) & 1 & 0 \\
\e(xz) & \e(x^2z) & 0 & \e(x^2)
\end{pmatrix}
=
\sigma^{-2}
\begin{pmatrix}
D & C \\
C & D
\end{pmatrix},
\end{equation}
where
\[
C =
\begin{pmatrix}
\e(z) & \e(xz) \\
\e(xz) & \e(x^2z)
\end{pmatrix}
\quad\text{and}\quad
D =
\begin{pmatrix}
1 & 0 \\
0 & \e(x^2)
\end{pmatrix}
\]
and we have omitted the dependence of the expectations on the design $p$ for brevity. 
We emphasize that $\ci$ depends on $F$ as well, 
though the experimenter can only control $p$.
Furthermore, when $F = (1/n)\sum_{i=1}^n \delta_{x_i}$
and the running variable values $x_1,\dots,x_n$ are mean-centered,
the expected information matrix $\ci(p)$ is precisely the fixed-$x$ information matrix $\ci_n(p_1,\ldots,p_n)$,
identifying $p_i \equiv p(x_i)$.
This shows that indeed, the random-$x$ problem~\eqref{eq:p_opt_problem} 
is strictly more general than the fixed-$x$ problem~\eqref{eq:p_opt_problem_fd}.
Equation~\eqref{eq:I} also shows
that any efficiency objective $\Psi(\ci(p))$ only depends on the treatment indicators $z$ through their marginal distributions conditional on $x$, and not their
joint distribution.  
In the fixed-$x$ setting,
this makes it easier to obey an exact
budget constraint $n^{-1}\sum_{i=1}^n z_i = \ol{p}$ by stratification.  
For instance, given five subjects with $p_i=0.4$ we could randomly treat exactly
two of them,
instead of randomizing each subject independently
and possibly going over budget.

While we will characterize solutions to the optimal design problem~\eqref{eq:p_opt_problem} for any continuous efficiency criterion $\Psi(\cdot)$,
in Section~\ref{sec:inadmissibility} we will prove some additional results for the $D$-optimality criterion
$\Psi_D(\cdot) = \log(\det(\cdot))$.
We show that $D$-optimality is of particular interest in this setting,
as it happens to correspond exactly with the $c$-optimality efficiency criterion of~\citet{owen:vari:2020}.
They aim to minimize the asymptotic variance of $\hat{\beta}_3$ and do so by observing that
if $\ci$ is invertible,
then when $(x_i,z_i)$ are independent, 
by the law of large numbers
\[
n\Var(\hat{\beta}\giv \cx) = n (\cx^\tran\cx)^{-1}\, \toas\, \ci^{-1}.
\]
We have assumed $\sigma^2=1$ WLOG as the $D$-optimal design does not depend on $\sigma^2$.
Then by standard block inversion formulas
\begin{equation}
\label{eq:var_beta_3}
n \Var(\hat{\beta}_3\giv \cx) \toas  (\ci^{-1})_{44} = 
\frac{M_{11}(p)}{\det(M(p))}
\end{equation}
where
\begin{align}
\label{eq:M}
M = M(p) &= D-CD^{-1}C \notag\\ 
&=
\begin{pmatrix}
1 - \e(z)^2 - \frac{\e(xz)^2}{\e(x^2)} & -\e(xz) \cdot \e(z)-\frac{\e(x^2z)\e(xz)}{\e(x^2)} \\[1.5ex]
-\e(xz) \cdot \e(z)-\frac{\e(x^2z)\e(xz)}{\e(x^2)} & \e(x^2)-\e(xz)^2 - \frac{\e(x^2z)^2}{\e(x^2)}
\end{pmatrix}.
\end{align}
Equation~\eqref{eq:var_beta_3} shows that minimizing the asymptotic conditional variance of $\hat{\beta}_3$ is equivalent to maximizing $\eff(p) := \det(M(p))/M_{11}(p)$,
which under the present formalization of the problem is further equivalent to $c$-optimality for the expected information matrix~\eqref{eq:I} with $c=(0,0,0,1)^{\top}$.
The following result shows that $M_{11}(p)>0$ for any input-feasible design $p$.
It follows that $\eff(p)$ is always well-defined and nonnegative
for any input-feasible design.

\begin{corollary}
\label{corollary:M11}
For any $(\ol{z},\ol{xz}) \in \cj$, $M_{11} = 1-\ol{z}^2-(\ol{xz})^2/\e(x^2) > 0$.
\end{corollary}
\begin{proof}
See Appendix \ref{app:M11}.
\end{proof}

On the other hand,
because $\ci$ is the expected value of a positive semi-definite
rank one matrix,
it is also positive semi-definite.
Thus
\[
0 \le \det(\ci) = \det(D)\det(M) = \e(x^2)\det(M).
\]
This shows that $\det(M) \ge 0$ with inequality iff $\ci$ is invertible. 
Additionally, since $M_{11}$ only depends on $p$ through $\e_p(z)$ and $\e_p(xz)$, 
any two input-feasible designs $p$ and $p'$ satisfying the equality constraints in~\eqref{eq:p_opt_problem}
must have $M_{11}(p)=M_{11}(p')$.
It follows that the solutions to~\eqref{eq:p_opt_problem} under the $D$-optimality criterion $\Psi(\ci(p)) = \Psi_D(\ci(p))$
and the $c$-optimality criterion $\Psi(\ci(p))=\eff(p)$
must be identical
whenever $(\ol{z},\ol{xz}) \in \cj$.

\section{Optimal design characterizations}
\label{sec:characterizations}
To solve the constrained optimization problem~\eqref{eq:p_opt_problem},
we begin by observing that the expected information matrix $\ci(p)$,
computed in~\eqref{eq:I},
only depends on the design function $p$
through the quantities $\e_p(z)$, $\e_p(xz)$, and $\e_p(x^2z)$.
Then the same is true for any efficiency objective $\Psi(\ci(p))$.
Consequently for any continuous $\Psi$ we can write
$\Psi(\ci(p))) = g_{\Psi}(\e_p(z),\e_p(xz),\e_p(x^2z))$
for some continuous $g_{\Psi}:\real^3 \to \real$
that may depend on the running variable distribution $F$.

Fixing $(\ol{z},\ol{xz}) \in \cj$
and the set $\cf$ of permissible design functions,
we say a \textit{feasible} design $p \in \cf$ is one that satisfies the equality constraints in~\eqref{eq:p_opt_problem},
i.e. $\e_p(z)=\ol{z}$ and $\e_p(xz)=\ol{xz}$.
Thus, the efficiency criterion $\Psi(\ci(p))$
can only vary among feasible designs $p$ through the single quantity $\e_p(x^2z)$.
Furthermore, any two feasible designs $p$ and $q$
with $\e_p(x^2z)=\e_q(x^2z)$
must have the same efficiency.
Thus, we can break down the problem~\eqref{eq:p_opt_problem}
into two steps.
First, we find a solution
\begin{equation}
\label{eq:argmax}
\ol{x^2z}^*(\ol{z},\ol{xz};\Psi) \in \argmax_{a \in I_{\cf}(\ol{z},\ol{xz})} g_{\Psi}(\ol{z}, \ol{xz}, a)
\end{equation}
where 
\begin{equation}
\label{eq:I_F}
I_{\cf}(\ol{z},\ol{xz}) = \{\e_p(x^2z) \mid \e_p(z)=\ol{z},\e_p(xz)=\ol{xz} \text{ for some $p \in \cf$}\} \subseteq [-\e(x^2), \e(x^2)]
\end{equation}
is the set of values of $\e_p(x^2z)$ attainable by some feasible design $p \in \cf$.
Then we must find a feasible design $p \in \cf$
that satisfies $\e_p(x^2z)=\ol{x^2z}^*(\ol{z},\ol{xz};\Psi)$.

The next result shows that when $\cf$ is convex,
$I_{\cf}(\ol{z},\ol{xz})$ is an interval.
For our two choices of $\cf$ of interest,
Propositions~\ref{prop:pmax_pmin} and~\ref{prop:pmax_pmin_prime} will show it is a closed interval,
so~\eqref{eq:argmax} will always have a solution when $\Psi(\cdot)$ is continuous.

\begin{lemma}
\label{lemma:convex_combinations}
Suppose feasible designs $p$, $p' \in \cf$ satisfy $\e_{p}(x^2z) \le \e_{p'}(x^2z)$,
where $\cf$ is convex. 
Then if $\e_{p}(x^2z) \le \gamma \le \e_{p'}(x^2z)$,
there exists feasible $p^{(\gamma)} \in \cf$ with $\e_{p^{(\gamma)}}(x^2z) = \gamma$.
\end{lemma}
\begin{proof}
If $\e_{p'}(x^2z)=\e_{p}(x^2z)$ then either of them is a suitable $p^{(\gamma)}$. 
Otherwise take $\lambda \in[0,1]$ so that
$\gamma = \lambda \e_{p}(x^2z) + (1-\lambda) \e_{p'}(x^2z)$.
Then $p^{(\gamma)} = \lambda p+ (1-\lambda) p'$ is in $\cf$ by convexity and,
by direct computation of the moments
$\e_p(x^az)$ for $a \in \{0,1,2\}$,
feasible with $\e_{p}(x^2z)=\gamma$.
\end{proof}
The endpoints of $I_{\cf}(\ol{z},\ol{xz})$ can be computed as the optimal values of the following constrained optimization problems:
\begin{equation}
\label{eq:max_min_inequality}
\begin{array}{lllll}
\mbox{maximize}\quad & \e_p(x^2z) & \qquad 
&\mbox{minimize}\quad & \e_p(x^2z)  \\
\mbox{over} & p \in \cf & \quad  
&\mbox{over} & p \in \cf  \\
\mbox{subject to} & \e_p(z) = \ol{z} & \qquad & \mbox{subject to} & \e_p(z) = \ol{z}  \\
\mbox{and}& \e_p(xz) = \ol{xz} & \qquad &\mbox{and} & \e_p(xz) = \ol{xz}.
\end{array}
\end{equation}
Given solutions $p_{\max}$ and $p_{\min}$ to the problems~\eqref{eq:max_min_inequality},
Lemma~\ref{lemma:convex_combinations} shows that the design
\begin{equation}
\label{eq:optimal_design}
p_{\opt}(x) =
\begin{cases}
p_{\max}(x), & \e_{p_{\max}}(x^2z) \leq \ol{x^2z}^{*}(\ol{z},\ol{xz};\Psi)\\
p_{\min}(x), & \e_{p_{\min}}(x^2z) \geq \ol{x^2z}^{*}(\ol{z},\ol{xz};\Psi)\\
\lambda p_{\min}(x) + (1-\lambda)p_{\max}(x), &   \text{else}
\end{cases}
\end{equation}
solves the problem~\eqref{eq:p_opt_problem} for $\lambda = (\e_{p_{\max}}(x^2z)- \ol{x^2z}^{*}(\ol{z},\ol{xz};\Psi))/(\e_{p_{\max}}(x^2z)-\e_{p_{\min}}(x^2z))$.
If $\e_{p_{\max}}(x^2z)=\e_{p_{\min}}(x^2z)$ then all feasible designs $p$ have the same efficiency, so any one of them is optimal.

The remainder of this section is concerned with characterizing the solutions to the problems~\eqref{eq:max_min_inequality}
for two specific choices of design function classes $\cf$: 
the set of all measurable functions into $[0,1]$,
and the set of all such \textit{monotone} functions.
For these two choices of $\cf$,
solutions $p_{\max}$ and $p_{\min}$ exist
for any $(\ol{z},\ol{xz}) \in \cj$
and are unique.
Our argument uses extensions of the Neyman-Pearson lemma~\citep{neymanpearson1933} in hypothesis testing.
These extensions are in~\cite{dantzigwald1951},
whose two authors discovered the relevant results independently of each other.
We use a modern formulation of their work,
adapting the presentation by~\citet{lehmannromano2005}:
\begin{lemma} \label{lemma:np}
Consider any measurable $h_1,\dots,h_{m+1}:\real \rightarrow \real$ 
with $\e(|h_i(x)|)<\infty, i=1,\dots,m+1$. 
Define $S \subseteq \real^m$ to be the set of all points $c=(c_1,\dots,c_m)$ such that 
\begin{equation}
\label{eq:np_constraints}
\e(p(x)h_i(x))=c_i, \quad i=1,\dots,m,
\end{equation}
for some $p \in \cf$, where $\cf$ is some collection of measurable functions from $\real$ into $[0,1]$.
For each $c \in S$ let $\cf_c$ be the set of all $p \in \cf$ satisfying~\eqref{eq:np_constraints}. If $\cf$ is such that 
\[
S' = \bigl\{(c,c_{m+1}) \in \real^{m+1}\mid c \in S,\  c_{m+1}=\e(p(x)h_{m+1}(x)) \text{ for some $p \in \cf_c$}\bigr\}
\] 
is closed and convex and $c \in \intr \ S$, then
\begin{enumerate}
    \item \label{lemma:np:existence} There exists $p \in \cf_c$ and $k_1,\ldots,k_m \in \real$ such that
    \begin{equation} \label{eq:np_um}
        p \in \argmax_{q \in \cf} \e\biggl(q(x)\biggl(h_{m+1}(x)-\sum_{i=1}^m k_ih_i(x)\biggr)\biggr),\quad\text{and}
    \end{equation}
    \item \label{lemma:np:suff_nec} $p \in \argmax_{q \in \cf_c} \e(q(x)h_{m+1}(x))$ if and only if $p \in \cf_c$ satisfies~\eqref{eq:np_um} for some $k_1,\dots,k_m$.
\end{enumerate}
\end{lemma}
\begin{proof}
Claim \ref{lemma:np:existence} and necessity of~\eqref{eq:np_um} in claim \ref{lemma:np:suff_nec} follows from the proof of part (iv) of Theorem 3.6.1 in \citet{lehmannromano2005}, which uses the fact that $S'$ is closed and convex to construct a separating hyperplane in $\real^{m+1}$. Sufficiency of~\eqref{eq:np_um} in claim \ref{lemma:np:suff_nec} follows from part (ii) of that theorem, and is often called the method of undetermined multipliers.
\end{proof}
Lemma~\ref{lemma:np} equates a constrained optimization problem (item~\ref{lemma:np:suff_nec})
and a compound optimization problem (item~\ref{lemma:np:existence}).
Unlike typical equivalence theorems,
it does not require $\cf$ to be the set of all measurable design functions,
and uses an entirely different proof technique.
Following~\citet{whittle1973},
equivalence theorems in optimal design are now popularly proven 
using the concept of Fr\'{e}chet derivatives
on the space of design functions (measures).
However,
such approaches often do not apply when $\cf$ is restricted to be the set of all monotone design functions.
Most relevant to our problem, the proof of Corollary 2.1 in the supplement of~\citet{metelkinapronzato2017}
involves Fr\'{e}chet derivatives
in the direction of design functions supported at a single value of $x$,
which are not monotone.
However, the use of Lemma~\ref{lemma:np} requires an objective linear in $p$,
where typical equivalence theorems only require concavity.

\subsection{Globally optimal designs}
\label{sec:global_optimal}
We now solve the design problem~\eqref{eq:p_opt_problem} in the case that $\cf$ is the set of 
all measurable functions $p:\real\to[0,1]$.
We first explain how the results of~\citet{metelkinapronzato2017} do not adequately do so already.
Identifying our design functions $p$ with their design measures $\xi$,
Corollary 2.1 of~\citet{metelkinapronzato2017}
does not provide any information about what an optimal solution $p_{\opt}(x)$ to~\eqref{eq:p_opt_problem} would be for values of $x$ where $G_1(p_{\opt}(\cdot);x)=G_2(p_{\opt}(\cdot);x)$.
Here $G_1$ and $G_2$ are quantities derived from the aforementioned Fr\'{e}chet derivatives,
depending on the constraints $\ol{z}$ and $\ol{xz}$.
Unfortunately, this lack of information about $p_{\opt}$ holds for all $x$ both in their Example 1 and in our setting.
Their 
example skirts this limitation by noting some moment conditions on $p_{\opt}$ implied by the equality $G_1=G_2$
when the running variable is uniform and the efficiency criterion is $D$-optimality,
and then manually searching for some parametric forms of $p_{\opt}$ for which it is possible to satisfy these conditions.
By contrast,
the results in this section apply Lemma~\ref{lemma:np} with $\cf$ the set of all design functions,
and show a simple
stratified design function
is always optimal for any running variable distribution $F$ and continuous efficiency criterion.
This enables optimal designs to be systematically and efficiently constructed
(Section~\ref{sec:examples}).

We will apply Lemma~\ref{lemma:np} with $m=2$ constraints pertaining to $h_1(x)=1$ and $h_2(x)=x$. 
Our objective function is based on $h_3(x)=x^2$.
When the running variable distribution $F$ is continuous,
recalling the notation~\eqref{eq:p_A}
the solutions to~\eqref{eq:max_min_inequality} take the forms $p_{\max} = p_{[a_1,a_2]^c}$ and $p_{\min}=p_{[b_1,b_2]}$ for some intervals $[a_1,a_2]$ and $[b_1,b_2]$.
\begin{proposition} 
\label{prop:pmax_pmin}
Let $\cf$ be the set of all measurable functions from $\real$ into $[0,1]$.
For any $(\ol{z},\ol{xz}) \in \cj$, there exist unique solutions $p_{\max}$ and $p_{\min}$ to the optimization problems~\eqref{eq:max_min_inequality}.
These solutions are the unique feasible designs satisfying
\begin{equation}\label{eq:gen_p_max_min}
p_{\max}(x)=
\begin{cases}
1, & x \not\in[a_1,a_2]\\
0, & x\in(a_1,a_2)
\end{cases} 
\qquad \text{and}\qquad
p_{\min}(x) = 
\begin{cases}
1, & x\in(b_1,b_2) \\
0, & x \not\in[b_1,b_2]
\end{cases}
\end{equation}
for some $a_1 \le a_2$ and $b_1 \le b_2$ which depend on $(\ol{z},\ol{xz})$ and can be infinite if $\ol{xz}=\ol{xz}_{\max}(\ol{z})$.
\end{proposition}
\begin{proof}
If $\ol{xz}=\ol{xz}_{\max}(\ol{z})$, then the proposition follows by Lemma~\ref{lemma:stg} and taking
$a_1=-\infty$, $a_2=t=b_1$ and $b_2=\infty$. 
Thus we can assume that $\ol{xz} < \ol{xz}_{\max}(\ol{z})$.
We give the proof for $p_{\max}$ in detail. The argument for $p_{\min}$ is completely symmetric. 

As noted above, we are in the setting of Lemma~\ref{lemma:np} with $m=2$, $h_1(x)=1$, $h_2(x)=x$, and $h_3(x)=x^2$.
The collection $\cf$ here is the set of all measurable functions from $\real$ into $[0,1]$,
so the corresponding $S'$ is closed and convex, as shown in part (iv) of Theorem 3.6.1 in~\citet{lehmannromano2005}. 
By Lemma~\ref{lemma:stg} and~\eqref{eq:tower} we can write $\intr \ S = \varphi_1(\ct)$ 
where $\varphi_1$ is defined in the discussion around~\eqref{eq:tower} and
\[
\ct = \{(\ol{z},\ol{xz}) \mid -1 < \ol{z} < 1, \ol{xz}_{\min}(\ol{z}) < \ol{xz} < \ol{xz}_{\max}(\ol{z})\}
\]
Hence our previous assumption $\ol{xz} < \ol{xz}_{\max}(\ol{z})$ ensures $c=\varphi_1(\ol{z},\ol{xz}) \in \intr \ S$.

With the conditions of Lemma~\ref{lemma:np} satisfied,
we now show that~\eqref{eq:np_um} is equivalent to~\eqref{eq:gen_p_max_min} for any feasible $p_{\max}$. 
A feasible design $p_{\max}$ satisfies~\eqref{eq:np_um} iff $p_{\max} \in \argmax_{q \in \cf} \e(q(x)(x^2-k_1x-k_2))$ for some $k_1,k_2$,
or equivalently
\begin{equation} \label{eq:one_zero}
p_{\max}(x) = 
\begin{cases}
1, & x^2-k_1x-k_2 > 0 \\
0, & x^2-k_1x-k_2 < 0
\end{cases}
\end{equation}
(cf. part (ii) of Theorem 3.6.1 in~\citet{lehmannromano2005}). 
If $x^2-k_1x-k_2$ has no real roots then $p_{\max}(x)=1$ for all $x$, contradicting $\ol{z}<1$. Thus we write $x^2-k_1x-k_2=(x-a_1)(x-a_2)$ for some (real) $a_1 \le a_2$, showing that~\eqref{eq:one_zero} is equivalent to~\eqref{eq:gen_p_max_min}.
We can now conclude, by the second claim in Lemma~\ref{lemma:np}, 
that the set of optimal solutions to~\eqref{eq:max_min_inequality} 
contains precisely those feasible designs satisfying~\eqref{eq:gen_p_max_min}.
Furthermore, the first claim of Lemma~\ref{lemma:np} ensures that such a design must exist.
It remains to show only one feasible design can satisfy~\eqref{eq:gen_p_max_min};
in Appendix~\ref{app:uniqueness} we provide a direct argument, which does not rely on Lemma~\ref{lemma:np}.

\end{proof}
\begin{remark}
The necessity and sufficiency results of Proposition~\ref{prop:pmax_pmin} do
follow from Corollary 2.1 of~\citet{metelkinapronzato2017}.
We again identify our design functions $p$ with their design measures $\xi$
and take $\psi(\xi) = \int x^2 \rd\xi(x)$,
which can be written as an affine function $\bm{\Psi}(\cdot)$
of the expected information matrix $\ci(\xi)$.
As discussed at the beginning of this section, however, the form of a solution to problem~\eqref{eq:p_opt_problem},
cannot be constructed from their Corollary
without the reduction to~\eqref{eq:max_min_inequality}
and applying~\eqref{eq:optimal_design},
so that $\psi$ is as above rather than
something like $D$-optimality.
We have also shown a stronger uniqueness result than Section 2.3.3 of~\citet{metelkinapronzato2017},
which only applies when the running variable distribution $F$ has a density with respect to Lebesgue measure.
Our Lemma~\ref{lemma:np} also provides an existence guarantee that does not rely on strict concavity of $\bm{\Psi}(\cdot)$
on the set of positive definite matrices;
this is violated by the affine choice we need here.
\end{remark}

As we will see in Section~\ref{sec:examples}, 
when $\ol{z} < 0$ we frequently encounter $p_{\opt}=p_{\max}$
under $D$-optimality.
In this case it is intuitive that there is an
efficiency advantage to strategically allocate the rare level $z=1$ at both high and low $x$, compared to a three level tie-breaker.
But such a design is usually unacceptable in our motivating problems.
We will thus constrain $\cf$ to the set of monotone design functions in Section~\ref{sec:monotonicity}.

Before doing that, we present an alternative solution to~\eqref{eq:p_opt_problem}
assuming the running variable distribution has an additional moment.
This result shows
that when the running variable is continuous,
an optimal design with no randomization always exists.
However, randomized assignment
becomes essential once we restrict our attention to monotone designs in Section~\ref{sec:monotonicity},
as the only non-randomized monotone designs
are generalized RDDs (Remark~\ref{remark:generalized_RDD}).
\begin{theorem}
\label{thm:alt_pmax_pmin}
Suppose $\e(|x|^3) < \infty$.
Then when $\cf$ is the set of all measurable design functions,
for any $(\ol{z},\ol{xz}) \in \cj$
there exists a solution to~\eqref{eq:p_opt_problem} with 
\begin{equation}
\label{eq:p_opt} 
p_{\opt}(x) =
\begin{cases}
1, & x < a_1 \\
0, & a_1 < x < a_2 \\
1, & a_2 < x < a_3 \\
0, & x > a_3
\end{cases}
\end{equation}
for some $a_1 \le a_2 \le a_3$ which are finite 
unless $p_{\opt}$ is one of the designs $p_{\max}$ or $p_{\min}$ 
in~\eqref{eq:gen_p_max_min}.
\end{theorem}
\begin{proof}
The solutions to~\eqref{eq:p_opt_problem}
are precisely the feasible design functions $p \in \cf$ where $\e_p(x^2z)$ is a solution to~\eqref{eq:argmax}.
Fix any such solution $\ol{x^2z}^*(\ol{z},\ol{xz};\Psi)$ to~\eqref{eq:argmax};
it suffices to find a feasible design $p_{\opt}$
with $\e_{p_{\opt}}(x^2z) = \ol{x^2z}^*(\ol{z},\ol{xz};\Psi)$.
If $\ol{x^2z}^*(\ol{z},\ol{xz};\Psi)$ is the lower (resp. upper) endpoint of the interval $I_{\cf}(\ol{z},\ol{xz})$,
then by Proposition~\ref{prop:pmax_pmin},
the unique feasible design with $\e_p(x^2z)=\ol{x^2z}^*(\ol{z},\ol{xz};\Psi)$ 
is the design $p_{\min}$ (resp. $p_{\max}$).
Then the result follows
with $a_1=-\infty$ (resp. $a_3=\infty$).

Otherwise, $\ol{x^2z}^*(\ol{z},\ol{xz};\Psi)$ is in the interior of $I_{\cf}(\ol{z},\ol{xz})$,
and we aim to apply Lemma~\ref{lemma:np} with $m=3$, $h_i(x)=x^i$ for $i \in \{1,2,3\}$,
and $c=\varphi_2(\ol{z},\ol{xz},\ol{x^2z}^*(\ol{z},\ol{xz};\Psi))$.
With $S'$ closed and convex as shown in Proposition~\ref{prop:pmax_pmin},
we only need to show $c \in \intr\ S$.
With the interior of $I_{\cf}(\ol{z},\ol{xz})$ being nonempty,
there is more than one feasible design and
the uniqueness result of Lemma~\ref{lemma:stg} indicates that we must have $\ol{xz} < \ol{xz}_{\max}(\ol{z})$.
With $\ol{z} \in (-1,1)$ by assumption and
$\ol{x^2z}^*(\ol{z},\ol{xz};\Psi) \in \intr\ I_{\cf}(\ol{z},\ol{xz})$,
indeed $c \in \intr \ S$.

Applying Lemma~\ref{lemma:np} we know that there exists a feasible design $p_{\opt}$
with $\e_{p_{\opt}}(x^2z) = \ol{x^2z}^*(\ol{z},\ol{xz};\Psi)$ and
$p_{\opt}(x) \in \argmax_{q \in \cf} \e(q(x)(-x^3-k_1x^2-k_2x-k_3))$
for some $k_1,k_2,k_3$.
If $f(x;k_1,k_2,k_3) \equiv -x^3-k_1x^2-k_2x-k_3$ had only one real root $a_1$,
then the negative leading coefficient indicates $f(x;k_1,k_2,k_3) > 0$ when $x<a_1$
and $f(x;k_1,k_2,k_3) < 0$ when $x>a_1$.
This would imply $p_{\opt}(x)$ is a design that always treats all subjects with $x<a_1$
and never treats any subject with $x>a_1$,
which cannot be input-feasible.
We conclude $f(x;k_1,k_2,k_3)=-(x-a_1)(x-a_2)(x-a_3)$ for some finite
$a_1 \le a_2 \le a_3$
which are the roots of $f(x;k_1,k_2,k_3)$.
This shows the existence of $p_{\opt}$ of the form~\eqref{eq:p_opt}.
\end{proof}

\subsection{Imposing a monotonicity constraint}\label{sec:monotonicity}
We now apply Lemma~\ref{lemma:np} to solve~\eqref{eq:p_opt_problem}
in the case of principal interest,
where $\cf$ is the set of all \textit{monotone} design functions.
Note that
the lower bound $\ol{xz} \geq 0$ that we
imposed in Section \ref{subsection:stg_bounds} does not exclude
any monotone designs.  If $p(x)$ is monotone
then $x$ and $z$ necessarily have a nonnegative covariance $\e_p(xz)$.

Our argument follows the outline of Section~\ref{sec:global_optimal}.
Suppose that $p_{\max}^{\dag}$ and $p_{\min}^{\dag}$ are solutions to~\eqref{eq:max_min_inequality} with $\cf$ the set of monotone design functions,
which we distinguish from the optimal designs $p_{\max}$ and $p_{\min}$ of Section~\ref{sec:global_optimal}.
As $\cf$ is convex, Lemma~\ref{lemma:convex_combinations} applies, and thus a solution $p_{\opt}^{\dag}$ to~\eqref{eq:p_opt_problem}
is given by~\eqref{eq:optimal_design},
replacing $p_{\max}$ and $p_{\min}$ by $p_{\max}^{\dag}$ and $p_{\min}^{\dag}$,
respectively.
Note that $\ol{x^2z}^*(\ol{z},\ol{xz})$
may differ from its value in Section~\ref{sec:global_optimal}
since $\cf$ has changed.

We now characterize the designs $p_{\max}^{\dag}$ and $p_{\min}^{\dag}$.
As in Proposition~\ref{prop:pmax_pmin},
these designs always exist and are unique
for any $(\ol{z},\ol{xz}) \in \cj$.
When $F$ is continuous, these are monotone two level designs
$p_{\max}^{\dag}=p_{\ell,1,t}$ and $p_{\min}^{\dag}=p_{0,u,s}$
as defined in~\eqref{eq:twolevel}.
For general $F$ the designs $p_{\max}^{\dag}$ and $p_{\min}^{\dag}$ may differ from these designs at at the single discontinuity.


\begin{proposition}
\label{prop:pmax_pmin_prime}
For any $(\ol{z},\ol{xz}) \in \cj$, there exist unique solutions $p_{\max}^{\dag}$ and $p_{\min}^{\dag}$ to the optimization problems~\eqref{eq:max_min_inequality},
when $\cf$ is the set of all monotone design functions.
These solutions are the unique feasible designs satisfying
\begin{equation}\label{eq:gen_p_max_min_mon}
p_{\max}^{\dag}(x)=
\begin{cases}
\ell, & x < t \\
1, & x > t
\end{cases} 
\qquad \text{and}\qquad
p_{\min}^{\dag}(x) = 
\begin{cases}
0, & x < s \\
u, & x > s
\end{cases}
\end{equation}
for some $\ell,u \in [0,1]$ and constants $s,t$, which all depend on $(\ol{z},\ol{xz})$, where $s$ and $t$ may be infinite if $\ol{xz}=0$.
\end{proposition}
\begin{proof}
If $\ol{xz}=0$, the only feasible monotone design is the fully randomized design $p(x)=(1+\ol{z})/2$, 
and the theorem holds trivially with $p_{\max}^{\dag}=p_{\min}^{\dag}$, $t=-s=\infty$, and $\ell=u=(1+\ol{z})/2$. 
Likewise, if $\ol{xz}=\ol{xz}_{\max}(\ol{z})$ then the desired results follow by Lemma~\ref{lemma:stg} 
(take $\ell=0$, $u=1$, and $s=t$ with $t$ as in Lemma~\ref{lemma:stg}). 
Thus, we assume that $0 < \ol{xz} < \ol{xz}_{\max}(\ol{z})$.
Again, we only write out the argument for $p_{\max}^{\dag}$;
the proof for $p_{\min}^{\dag}$ is completely analogous.

Once again, we are in the setting of Lemma~\ref{lemma:np} with $m=2$, $h_1(x)=1$, $h_2(x)=x$, and $h_3(x)=x^2$.
The only difference from Proposition~\ref{prop:pmax_pmin} is the definition of $\cf$, 
so we must verify that the conditions on the corresponding $S'$ and $S$ are satisfied.
Since $\e(x^ap(x))$ is linear in $p$ for all $a$,
and any convex combination of monotone functions is monotone (cf. Lemma~\ref{lemma:convex_combinations}),
$S'$ is convex.
Now suppose that $c_0$ is a limit point of $S'$. 
Then there exists a sequence $p_1,p_2,\dots\in \cf$ 
with $(\e(p_n(x)),\e(xp_n(x)),\e(x^2p_n(x))) \to c_0$ as $n \to \infty$. 
As $\cf$ is sequentially compact, there exists a subsequence $p_{n_i}$ and $p \in \cf$ with $p_{n_i} \to p$ pointwise. 
But then $(\e(p(x)),\e(xp(x)),\e(x^2p(x)))=c_0$ by dominated convergence, so $S'$ is closed. 
Finally, $\intr\ S=\varphi_1(\ct^{\dag})$ where 
\[
\ct^{\dag} \equiv \{(\ol{z},\ol{xz}) \mid -1 < \ol{z} < 1, 0 < \ol{xz} < \ol{xz}_{\max}(\ol{z})\} = \intr\ \cj.
\]
Hence the assumption that 
$0 < \ol{xz} < \ol{xz}_{\max}(\ol{z})$ ensures that $c=\varphi_1(\ol{z},\ol{xz}) \in \intr\ S$.

With the conditions of Lemma~\ref{lemma:np} once again satisfied, 
we now show that~\eqref{eq:np_um} is equivalent to~\eqref{eq:gen_p_max_min_mon} for any (monotone) feasible $p_{\max}^{\dag}$. 
First assume feasible $p_{\max}^{\dag}$ satisfies~\eqref{eq:np_um}, i.e. $p_{\max}^{\dag} \in \argmax_{q \in \cf} \e(q(x)(x^2-k_1x-k_2))$ for some $k_1,k_2$.
The polynomial $x^2-k_1x-k_2$ having no real roots would mean this condition is equivalent to $p_{\max}^{\dag}(x)=1$, contradicting $\ol{z}<1$.
Hence we can factor $x^2-k_1x-k_2=(x-r)(x-t)$ for some $r \le t$.
Considering the sign of $(x-r)(x-t)$ and monotonicity of any $q \in \cf$ we see
\begin{align*}
\e(q(x)(x-r)(x-t)) & \le \e(q(r)(x-r)(x-t)\ind(x < r))  + \e(q(r)(x-r)(x-t) \ind(r \le x < t)) \\
&\phe\ + \e((x-r)(x-t)\ind(x \ge t)).     
\end{align*}
This inequality is strict unless
$q(x)=1$ for almost every $x > t$
and $$(q(r)-q(x))(x-r)(x-t)\ind(x<r)=(q(x)-q(r))(x-r)(x-t)\ind(r \le x < t)=0$$ with probability one under $F$, i.e., $q(x)=q(r)=\ell$ for some $\ell \in [0,1]$ and almost every $x < t$.
Therefore any design in $\argmax_{q \in \cf} \e(q(x)(x^2-k_1x-k_2))$ must satisfy the first condition in~\eqref{eq:gen_p_max_min_mon}.
Conversely, if a feasible, monotone $p_{\max}^{\dag}$ satisfies~\eqref{eq:gen_p_max_min_mon} 
then let $r_1 \le t$ be such that $g(r_1) = \e((x-r_1)(x-t)\indic(x<t))=0$. 
Such $r_1$ exists since assuming WLOG that $\Pr(x<t) > 0$, $g(\cdot)$ is continuous on $(-\infty, t]$
with $-\infty = \lim_{k_1 \downarrow -\infty} g(k_1) < 0 \le g(t)$. 
Considering the signs of $(x-r_1)(x-t)$ we get for any $p \in \cf$
\begin{align*}
\e((p_{\max}^{\dag}(x)-p(x))(x-r_1)(x-t)) & = \e((1-p(x))(x-r_1)(x-t)\indic(x \ge t)) \\
&\phe\ + \e((\ell-p(x))(x-r_1)(x-t)\indic(x<t)) \\
& \ge \e((1-p(x))(x-r_1)(x-t)\indic(x \ge t)) + (\ell-p(r_1))g(r_1) \\
& = \e((1-p(x))(x-r_1)(x-t)\indic(x \ge t)) \ge 0
\end{align*}
and so $p_{\max}^{\dag}$ satisfies~\eqref{eq:np_um} with $k_1=r_1+t$, and $k_2=-tr_1$.
The second claim in Lemma~\ref{lemma:np} 
then ensures that the set of optimal solutions to~\eqref{eq:max_min_inequality} consists of precisely those feasible, monotone designs satisfying~\eqref{eq:gen_p_max_min_mon}.
Such a design must exist by the first claim of Lemma~\ref{lemma:np}.
The remaining uniqueness claims are shown in Appendix~\ref{app:uniqueness_mon}.
\end{proof}

Analogous to Theorem~\ref{thm:alt_pmax_pmin},
if we assume the running variable has a third moment
then we have a solution to~\eqref{eq:p_opt_problem}
of a simpler form than~\eqref{eq:optimal_design}.
When the running variable $x$ is continuous,
this solution will be a two level design $p_{\ell',u',t'}$
for some $0 \le \ell' \le u' \le 1$.
In general,
when $x$ is not continuous,
we may need a different treatment probability at the discontinuity $t'$.
\begin{theorem}
\label{thm:alt_pmax_pmin_mon}
Suppose $\e(|x|^3)<\infty$.
Then when $\cf$ is the set of all monotone design functions,
for any $(\ol{z},\ol{xz}) \in \cj$
there exists a solution to~\eqref{eq:p_opt_problem}
with
\begin{equation}
\label{eq:p_opt_mon}
p_{\opt}^{\dag}(x) = 
\begin{cases}
\ell', & x < t' \\
u', & x > t'
\end{cases}
\end{equation}
for some $0 \le \ell' \le u' \le 1$ and $t' \in \real$.
\end{theorem}
\begin{proof}
The proof structure is similar to that of Theorem~\ref{thm:alt_pmax_pmin}.
Fix any solution $\ol{x^2z}^*(\ol{z},\ol{xz};\Psi)$ to~\eqref{eq:argmax}.
If it is an endpoint of $I_{\cf}(\ol{z},\ol{xz})$
then the unique solution to~\eqref{eq:p_opt_problem}
is $p_{\max}^{\dag}$ or $p_{\min}^{\dag}$
from Proposition~\ref{prop:pmax_pmin_prime},
which takes the form~\eqref{eq:p_opt_mon}
with $u'=1$ or $\ell'=0$, respectively.
Otherwise we apply Lemma~\ref{lemma:np}
with $m=3$, $h_i(x)=x^i$ for $i \in \{1,2,3\}$,
and $c=\varphi_2(\ol{z},\ol{xz},\ol{x^2z}^*(\ol{z},\ol{xz};\Psi))$.
The lemma applies since $S'$ is closed and convex from the proof of Proposition~\ref{prop:pmax_pmin_prime},
and $c \in \intr \ S$
since our assumption that $\ol{x^2z}^*(\ol{z},\ol{xz};\Psi) \in \intr \ I_{\cf}(\ol{z},\ol{xz})$
indicates there is more than one feasible design,
so $0 < \ol{xz} < \ol{xz}_{\max}(\ol{z})$.

Applying Lemma~\ref{lemma:np} we see that there exists a design $p_{\opt}^{\dag}$
that solves~\eqref{eq:p_opt_problem}
with $p_{\opt}^{\dag} \in \argmax_{q \in \cf} \e(q(x)f(x;k_1,k_2,k_3))$
for some $k_1,k_2,k_3$.
Here,
as in the proof of Theorem~\ref{thm:alt_pmax_pmin},
$f(x) = f(x;k_1,k_2,k_3) \equiv -(x^3+k_1x^2+k_2x+k_3)$.
We show this implies $p_{\opt}^{\dag}$ is of the form~\eqref{eq:p_opt_mon}
using the following claim. \\

\noindent
\underline{Claim:} Suppose $\sign(f(x)) = \indic(x < a) - \indic(x > a)$ w.p.1 for some $a \in \real$
and $\cf$ is the set of monotone design functions.
Then $p \in \argmax_{q \in \cf} \e(q(x)f(x))$
implies $p(x)=p(a)$ w.p.1. \\
\underline{Proof of claim:}
For any monotone design $p$ we can define $\tilde{p}(x)=\min(p(x),p(a))$ so that 
\begin{align*}
\e((\tilde{p}(x)-p(x))f(x) & = \e((\tilde{p}(x)-p(x))f(x)\indic(x \ge a)) = -\e((p(x)-p(a)) f(x)\indic(x \ge a))
\end{align*}
is nonnegative,
and zero iff $p(x)=p(a)$ for almost all $x \ge a$.
Similarly by considering $\tilde{p}(x)=\max(p(x),p(a))$,
we conclude $p(x)=p(a)$ for almost all $x < a$. \\

We notice that $f$
has either one real root $a_1$ or three real roots $a_1,a_2,a_3$.
If $a_1$ is the only root,
we know $f(x)<0$ when $x > a_1$ and
$f(x) > 0$ when $x<a_1$,
since the leading coefficient of $f$ is negative.
Thus we can apply the claim directly to show
that $p_{\opt}^{\dag}$ is constant,
in particular of the form~\eqref{eq:p_opt_mon}
with $\ell'=u'$.
If there are three real roots we show $p_{\opt}^{\dag}$ is of this form
with $t'=a_2$.
Let $F_<$ ($F_>$) be the conditional distribution of $x$ given $x < a_2$
($x>a_2$), so
\[
\e_{F}(q(x)f(x)) = \e_{F_<}(q(x)f(x))\Pr(x < a_2)  + \e_{F_>}(q(x)f(x))\Pr(x > a_2)
\]
We conclude the condition $p_{\opt}^{\dag} \in \argmax_{q \in \cf} \e_F(q(x)f(x;k_1,k_2,k_3))$
implies $p_{\opt}^{\dag}(x)=p_{\opt}^{\dag}(a_1)$ for almost all $x < a_2$
and $p_{\opt}^{\dag}(x)=p_{\opt}^{\dag}(a_3)$ for almost all $x>a_2$
by applying the claim twice
(once for $F_<$, once for $F_>$).
\end{proof}

In general, the optimal designs derived in Theorems~\ref{thm:alt_pmax_pmin} and~\ref{thm:alt_pmax_pmin_mon}
are not unique
when $\ol{x^2z}^*(\ol{z},\ol{xz},\Psi)$ is not on the boundary of $I_{\cf}(\ol{z},\ol{xz})$.
For example,
in nondegenerate cases the solution $p_{\opt}^{\dag}$ in~\eqref{eq:p_opt_mon}
typically has two levels,
while the solution in~\eqref{eq:optimal_design}
(with the monotonicity constraint)
will have three levels.
As another example,
the three-level tiebreaker found by~\citet{owen:vari:2020} to be optimal when $F$ is uniform
and $\ol{z}=0$
does not take the form~\eqref{eq:p_opt}
whenever $\ol{xz} < \ol{xz}_{\max}(0)$.
Conversely, Propositions~\ref{prop:pmax_pmin} and~\ref{prop:pmax_pmin_prime}
guarantee a unique optimal design when $\ol{x^2z}^*(\ol{z},\ol{xz},\Psi)$ is one of the endpoints of $I_{\cf}(\ol{z},\ol{xz})$.


\section{Exploration-exploitation trade-off}
\label{sec:inadmissibility}
As discussed in Section~\ref{sec:introduction},
\cite{owen:vari:2020} showed that when $\ol{z}=0$ and
$F \sim \dunif(-1,1)$, 
the efficiency (under their criterion $\eff(\cdot)$) of the three level tie-breaker~\eqref{eq:3_level_tie-breaker} 
is monotonically increasing in the width $\Delta$ of the randomization window. 
As $\Delta$ is a strictly decreasing function of $\ol{xz}$, 
and the three level tie-breaker solves~\eqref{eq:p_opt_problem} for all $\ol{xz}$, 
they conclude that there is a monotone trade-off between short-term gain and
statistical efficiency. 
In other words, greater statistical efficiency from an optimal design
requires giving up short-term gain.

We now extend these results to general $\ol{z}$ and other running variable distributions.
Hereafter $p_{\opt}=p_{\opt;\ol{z},\ol{xz}}$ denotes an optimal design without the monotonicity constraint,
to be contrasted with $p_{\opt}^{\dag}$ of Section~\ref{sec:monotonicity}.
Note we have made the dependence of these designs on $(\ol{z},\ol{xz}) \in \cj$ explicit.
We use the same efficiency criterion $\eff(\cdot)$ as~\citet{owen:vari:2020}.
Recall this is a $c$-optimality criterion
corresponding to the scaled asymptotic variance
of the OLS estimate for $\beta_3$ in~\eqref{eq:two_line_model},
and equivalent to $D$-optimality for our problem~\eqref{eq:p_opt_problem} by
Section~\ref{sec:efficiency}.

\begin{theorem}
\label{thm:opt_eff_monotonicity}
Suppose the distribution function $F$ of the running variable has a positive derivative everywhere in $I$, 
the smallest open interval with $\int_I f(x) \rd x = 1$.
If additionally $F(x)=1-F(-x),$ $\forall x \in I$, 
then fixing any $\ol{z} \in (-1,1)$, $\eff(p_{\opt;\ol{z},\ol{xz}})$ is decreasing in $\ol{xz}$.
\end{theorem}
\begin{proof}
See Appendix~\ref{app:opt_eff_monotonicity}.
\end{proof}
It turns out, however, that the gain versus efficiency trade-off is no longer monotone
under the monotonicity constraint.
Indeed, our next theorem shows that whenever $\ol{z} \neq 0$, if $F$ is symmetric (or indeed, not extremely skewed),
the fully randomized design $p_{\theta,\theta,0}$ is inadmissible for any $\theta \neq 1/2$, in the sense that there exists a different
\textit{monotone} design $p$ with $\e_p(z)=\ol{z}$ but both $\eff(p) > \eff(p_{\theta,\theta,0})$ and $\e_p(xz) > \e_{p_{\theta,\theta,0}}(xz)$.
In other words, the RCT is no longer admissible under $\eff(\cdot)$ when $\ol{z} \neq 0$. 

\begin{theorem}
\label{thm:inadmissibility}
Fix $\ol{z} \in (-1,1) \setminus \{0\}$,
and assume $F$ satisfies the conditions of Theorem~\ref{thm:opt_eff_monotonicity}.
If $\ol{z} < 0$ assume that $\e(x^2) < F^{-1}(1)^2$; otherwise assume that $\e(x^2) < F^{-1}(0)^2$.
Here $F^{-1}(1) \equiv \sup I$ and $F^{-1}(0) \equiv \inf I$.
Let $p_1=p_{\theta,\theta,0}$ be the fully randomized monotone design with $\e_{p_1}(z)=\ol{z}$, so that $\theta = (1+\ol{z})/2$.
Then there exists a monotone design $p_2$ such that $\e_{p_2}(z)=\ol{z}$ yet both
$\eff(p_{2}) > \eff(p_{1})$ and $\e_{p_{2}}(xz) > 0 = \e_{p_{1}}(xz)$.
\end{theorem}
\begin{proof}
See Appendix~\ref{app:inadmissibility}.
\end{proof}
\section{Examples}
\label{sec:examples}
In this section, we compute the optimal exploration-exploitation trade-off curves investigated in Section~\ref{sec:inadmissibility}
for several specific running variable distributions $F$. 
We can obtain large gains in efficiency under the criterion $\eff(\cdot)$
by moving from the three level tie-breaker design to $p_{\opt}^{\dag}$,
without sacrificing short-term gain.
We see further (generally smaller) improvements 
when we remove the monotonicity constraint and move from $p_{\opt}^\dag$ to $p_{\opt}$. 

To generate these curves we compute optimal designs $p_{\opt;\ol{z},\ol{xz}}$ and $p_{\opt;\ol{z},\ol{xz}}^{\dag}$
and evaluate their efficiency for various fixed $\ol{z} \in (-1,1)$
as we vary the short-term gain constraint $\ol{xz}$ over a fine grid covering $[0, \ol{xz}_{\max}(\ol{z})]$.
For interpretability we write $\ol{xz}=\delta \cdot \ol{xz}_{\max}(\ol{z})$ and specify short-term gain with the normalized parameter $\delta \in [0,1]$,
as discussed in Section~\ref{subsection:stg_bounds}.
When $F$ is continuous,
solutions $p_{\opt}$ and $p_{\opt}^{\dag}$ to~\eqref{eq:p_opt_problem} are computed by noting that we can write
$p_{\max}=p_{[a_1,a_2]^c}$, $p_{\min}=p_{[b_1,b_2]}$, $p_{\max}^{\dag}=p_{\ell,1,t}$, and $p_{\min}^{\dag}=p_{0,u,s}$
by Propositions~\ref{prop:pmax_pmin} and~\ref{prop:pmax_pmin_prime}.
Each of these designs has two unknown parameters
that must be the unique
solutions to the two feasibility constraints $\e_p(z)=\ol{z}$ and $\e_p(xz)=\ol{xz}$.
Given these parameters, we can apply~\eqref{eq:optimal_design} to compute $p_{\opt}$ and $p_{\opt}^{\dag}$.
If $\lambda \notin \{0,1\}$
we could also get an optimal design of the form in Theorem~\ref{thm:alt_pmax_pmin}. First we compute $\ol{x^2z}^*(\ol{z},\ol{xz};\eff)$
via~\eqref{eq:argmax},
noting that the endpoints of 
$I_{\cf}(\ol{z},\ol{xz})$
are $\e_{p_{\min}(\ol{z},\ol{xz}}(x^2z)$ and $\e_{p_{\max}(\ol{z},\ol{xz})}(x^2z)$.
Then~\eqref{eq:argmax} is simply maximizing a continuous function over a closed interval,
so it can be handled by standard methods such as Brent's algorithm~\citep{brent1973}.
Given $\ol{x^2z}^*(\ol{z},\ol{xz};\eff)$
we can then numerically search for $a_1$, $a_2$, $a_3$
such that $p_{\opt}=\indic (x \le a_1) + \indic(a_2 \le x \le a_3)$ is feasible with $\e_{p_{\opt}}(x^2z)=\ol{x^2z}^*(\ol{z},\ol{xz};\eff)$.
By Theorem~\ref{thm:alt_pmax_pmin} such a solution will exist and be optimal.
We can do a similar search for an optimal two level design $p_{\opt}^{\dag}$ under the monotonicity constraint,
by Theorem~\ref{thm:alt_pmax_pmin_mon}.

\subsection{Uniform running variable}
\label{sec:uniform}

\begin{table}
\caption{A list of the parameters for the various designs considered in Section \ref{sec:examples}, for fixed $(\ol{z}, \ol{xz}) \in \cj$, when $F \sim \dunif(-1,1)$. The values for $\Delta$ are only valid if they are between $0$ and $\min((1-\ol{z})/2, (1+\ol{z})/2)$, inclusive. Otherwise there is no feasible three level tie-breaker design.
}
\label{table:param_list}
\begin{tabular}{lcc}
\toprule
Design & Parameter & Value \\
\midrule
$p_{3;\ol{z},\Delta}$ & $\Delta$ & $2(1-\ol{z}^2-2\ol{xz})^{1/2}$ \\
\midrule
\multirow{2}{*}{$p_{\max;\ol{z},\ol{xz}}$} & $a_1$ & $-\ol{xz}/(1-\ol{z})-(1-\ol{z})/2$ \\
& $a_2$ & $-\ol{xz}/(1-\ol{z})+(1-\ol{z})/2$ \\
\midrule
\multirow{2}{*}{$p_{\min;\ol{z},\ol{xz}}$} & $b_1$ & $\ol{xz}/(1+\ol{z})-(1+\ol{z})/2$ \\
& $b_2$ & $\ol{xz}/(1+\ol{z})+(1+\ol{z})/2$ \\
\midrule
\multirow{2}{*}{$p_{\max;\ol{z},\ol{xz}}^{\dag}$} & $\ell$ & $(1/2)(1-\ol{z}^2-2\ol{xz})/(1-\ol{z}-\ol{xz})$ \\
 & $t$ & $1-2\ol{xz}/(1-\ol{z})$ \\
 \midrule
\multirow{2}{*}{$p_{\min;\ol{z},\ol{xz}}^{\dag}$} & $u$ & $(1/2)(1+z)^2/(1+\ol{z}-\ol{xz})$ \\
& $s$ & $2\ol{xz}/(1+\ol{z})-1$ \\
\bottomrule
\end{tabular}
\end{table}

We begin with the case $F \sim \dunif(-1, 1)$. 
This is the distribution most extensively studied by \cite{owen:vari:2020}, 
and allows closed form expressions for the parameters in $p_{\max}$, $p_{\min}$, $p_{\max}^{\dag}$, and $p_{\min}^{\dag}$,
given in Table~\ref{table:param_list}. 
Figure \ref{fig:uniform_efficiency} shows plots of $\eff(p)^{-1}$ versus $\ol{xz}$ for $\ol{z} \in \{0,-0.2,-0.5,-0.7\}$ 
under different designs:
the three level tie-breaker $p_{3;\ol{z},\Delta}$,
a globally optimal design $p_{\opt;\ol{z},\ol{xz}}$,
and
an optimal monotone design $p_{\opt;\ol{z},\ol{xz}}^{\dag}$.
Since $F$ is symmetric, 
the curves would be identical if $\ol{z}$ were replaced with $-\ol{z}$.

As shown in \citet{owen:vari:2020}, 
under the constraint $\ol{z}=0$ the three level tie-breaker is optimal for all $\delta$, 
and thus the three level tie-breaker, $p_{\opt}$, and $p_{\opt}^{\dag}$ all attain the optimal efficiency, 
as can be seen in the top left panel of Figure~\ref{fig:uniform_efficiency}. 
The proof of Theorem~\ref{thm:opt_eff_monotonicity}
shows this would hold for any continuous, symmetric running variable distribution $F$.
As $\ol{z}$ moves away from 0, however, we see that the three level tie-breaker becomes increasingly less efficient relative to both the optimal monotone design and the optimal design. 
At the same time, the range of short-term gain values $\ol{xz}$ attainable by three level tie-breaker designs becomes smaller relative to the full range achievable by arbitrary designs.
Note that Figure~\ref{fig:uniform_efficiency} plots the \textit{reciprocal} of the efficiency criterion $\eff(\cdot)$,
so that it can be interpreted
as an asymptotic variance for $\hat{\beta}_3$ via~\eqref{eq:var_beta_3},
and compared with the plots in~\citet{owen:vari:2020}.

\begin{figure}[t]
\centering  
\includegraphics[width=0.49\linewidth,height=6cm]{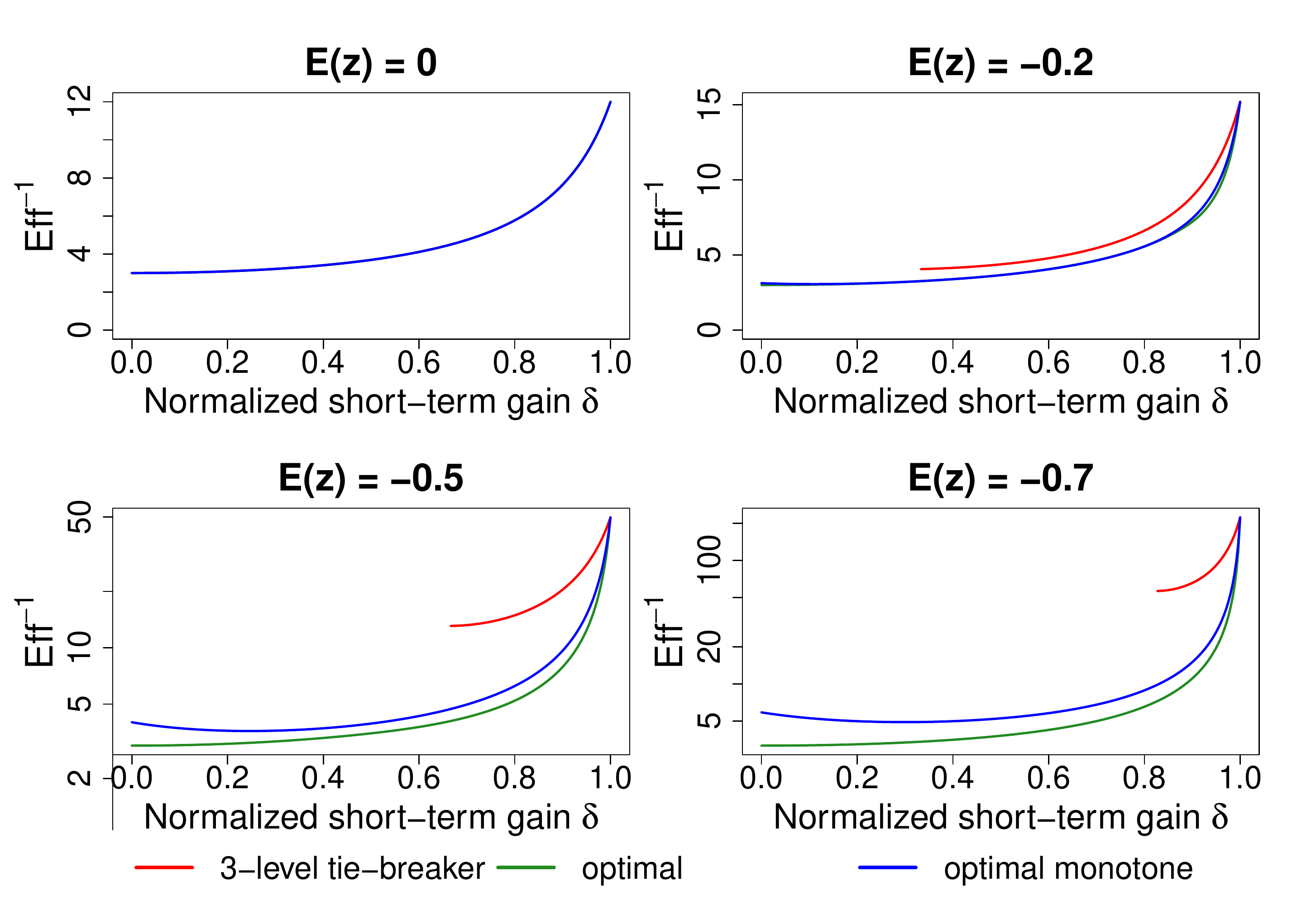} 
\includegraphics[width=0.49\linewidth,height=6cm]{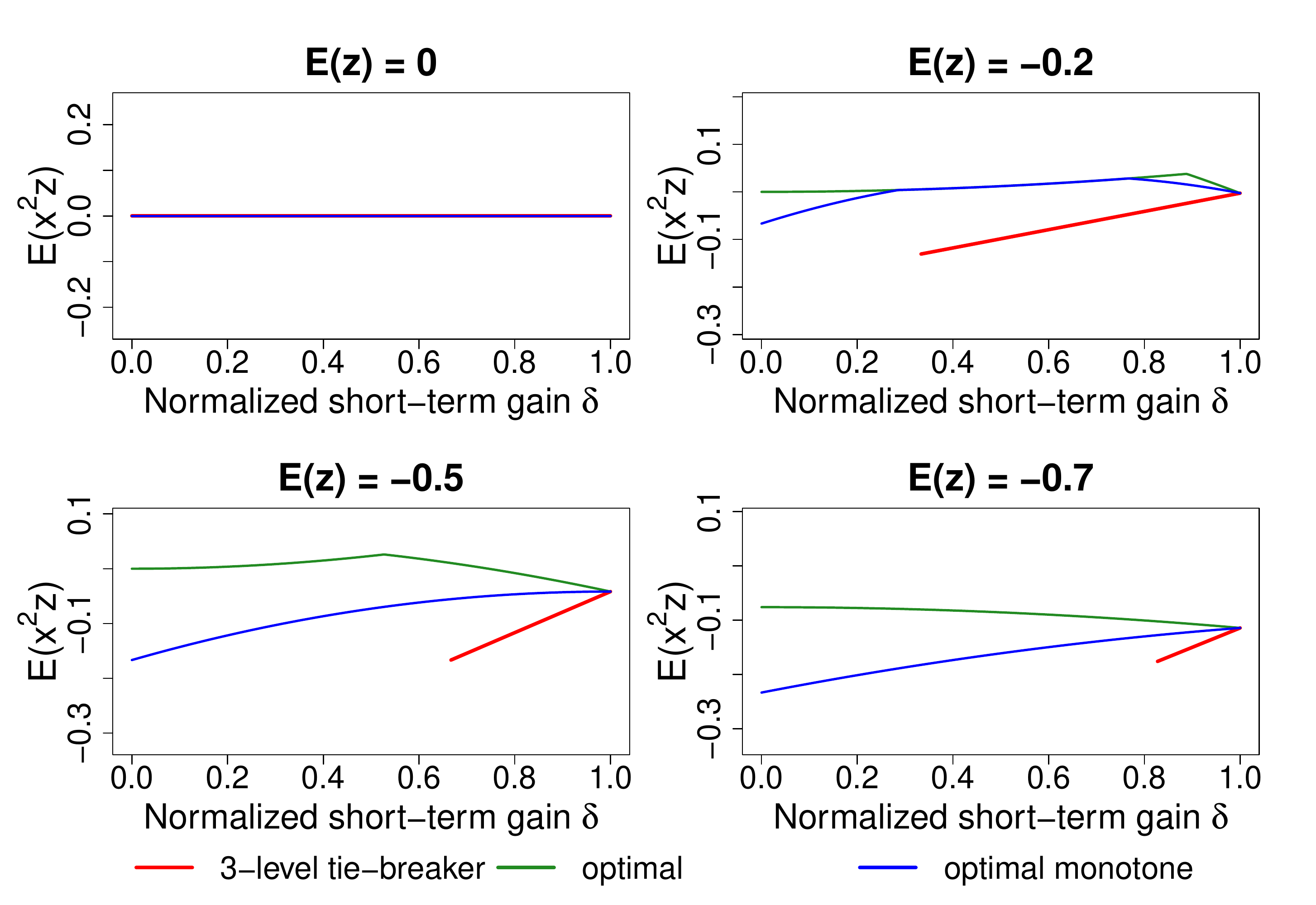} 
\caption{\label{fig:uniform_efficiency}
Left four plots: The inverse efficiencies of the three level tiebreaker $p_{3;\ol{z},\Delta}$,
the optimal designs $p_{\opt}$ in Section~\ref{sec:global_optimal},
and the optimal monotone designs $p_{\opt}^{\dag}$ in Section~\ref{sec:monotonicity}
as a function of the normalized short-term gain parameter $\delta$ 
for fixed $\ol{z} \in \{0, -0.2, -0.5, -0.7\}$, 
when $F \sim \dunif(-1,1)$.
When $\ol{z}=0$, all three curves coincide.
Note the logarithmic vertical axis for the bottom two plots.
Right four plots: The value of $\e_p(x^2z)$ for the designs $p$ in the left four plots.}
\end{figure}

\begin{table}[b]
\caption{
An extension of Table 2 in \citet{owen:vari:2020}, showing how $p_{\opt}$ and $p_{\opt}^{\dag}$ can greatly increase efficiency without sacrificing short-term gain, compared to a three level tie-breaker. All designs $p$ in this table satisfy $\e_p(z)=-0.7$ and assume $F \sim \dunif(-1,1)$.
}
\label{table:unif_example}
\begin{tabular}{llcc}
\toprule
Design $p$ & Description & Normalized short-term gain $\delta$ & $\eff(p)^{-1}$ \\
\midrule
$p_{0,1,0.7}$ &Sharp RDD & $1.000$ & $223.44$ \\
$p_{3;-0.7,\Delta}$ &3 level tie-breaker & $0.980\phz$ & $137.56$ \\
$p_{\opt;-0.7,0.25}^{\dag}$& Optimal monotone design & $0.980\phz$ & $\phz54.90$ \\
$p_{\opt;-0.7,0.25}$ &Optimal design & $0.980\phz$  & $\phz42.37$ \\
\bottomrule
\end{tabular}

\end{table}
Table~\ref{table:unif_example} extends Table 2 of \cite{owen:vari:2020}, 
referring to a setting in which only 15\% of subjects are to be treated ($\ol{z} = -0.7$).
That table shows the inverse efficiency of the sharp RDD $p_{0,1,0.7}$ is 223.44, 
while the three level tie-breaker $p_{3;-0.7,0.05}$ reduces this by about 40\% to $137.56$, at the cost of around 2\% of the short-term gain of the sharp RDD
over the RCT. 
Then $\eff(p_{\opt}^{\dag})^{-1}=54.90$ and $\eff(p_{\opt})^{-1}=42.37$, further improving efficiency
for designs $p_{\opt}$ and $p_{\opt}^{\dag}$ achieving the same short term gain as the three level tie-breaker. 
For this example, we can directly compute
with~\eqref{eq:quad_min} that $p_{\opt;-0.7,0.25}^{\dag} = p_{\max;-0.7,-0.25}^{\dag} = p_{\ell,1,t}$
is the unique optimal monotone design,
where by Table \ref{table:param_list},
$\ell=0.0034$ and $t = 0.7059$. 
In other words, the unique optimal monotone tie-breaker design deterministically assigns treatment to the top 14.7\%, 
and gives the other subjects an equal, small (0.34\%) chance of treatment. 

A limitation of this analysis is that in many practical settings, 
the two-line regression model will not fit very well over the entire range of $x$ values.
In that case the investigator might use a narrower data range,
essentially fitting a less asymmetric two-line model, as illustrated in \cite{owen:vari:2020}. 
This is equivalent to using a local linear regression with a rectangular ``boxcar'' kernel. 
In this setting, we know from Figure~\ref{fig:uniform_efficiency} that when the treatment proportion is not exactly 50\%, 
we can always do better than the three level tie-breaker using monotone two level design.
Even with a small asymmetry, e.g. 40\% treatment ($\ol{z} = -0.2$), 
we see a noticeable efficiency increase between the three level tie-breaker and an optimal monotone design across all values of $\delta$.

Finally, consistent with the results of Section~\ref{sec:inadmissibility}, 
we observe in Figure~\ref{fig:uniform_efficiency} that $\eff(p_{\opt})$ decreases with the gain parameter $\delta$ for each $\ol{z}$,
while near $\delta =0$, $\eff(p_{\opt}^{\dag})$ \textit{increases} with $\delta$ for all $\ol{z} \neq 0$. 
This clearly demonstrates the inadmissibility of the fully randomized design from Theorem \ref{thm:inadmissibility}. 
For example, if we fix $\ol{z}=-0.5$ (so 25\% of the subjects are to be treated), 
the fully randomized design $p_{0.25,0.25,0}=p_{\opt;-0.5,0}^{\dag}$ has efficiency 0.25 and no short-term gain ($\delta=0$), 
while $p_{\opt;-0.5,0.1}^{\dag}$ has higher efficiency (0.28)
with short-term gain $\delta = 0.27 > 0$. 
However, if we remove the monotonicity constraint, by Theorem~\ref{thm:opt_eff_monotonicity} $p_{\opt;-0.5,0}$ is the most efficient design over all attainable gain values, 
attaining efficiency 0.33 with $\delta=0$.


\subsection{Skewed running variable}
\begin{figure}
\centering
\includegraphics[width=0.49\linewidth,height=6cm]{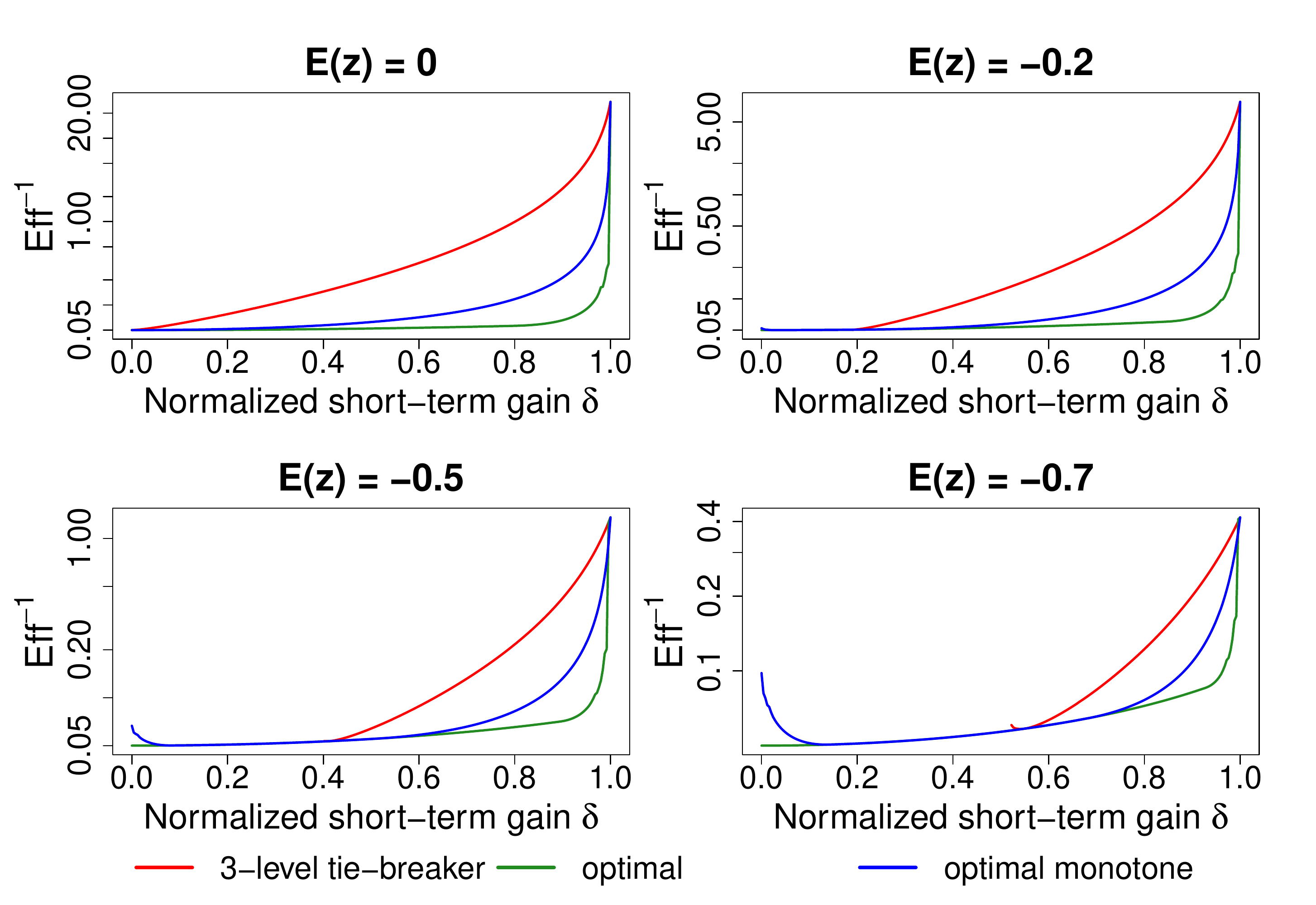} 
\includegraphics[width=0.49\linewidth,height=6cm]{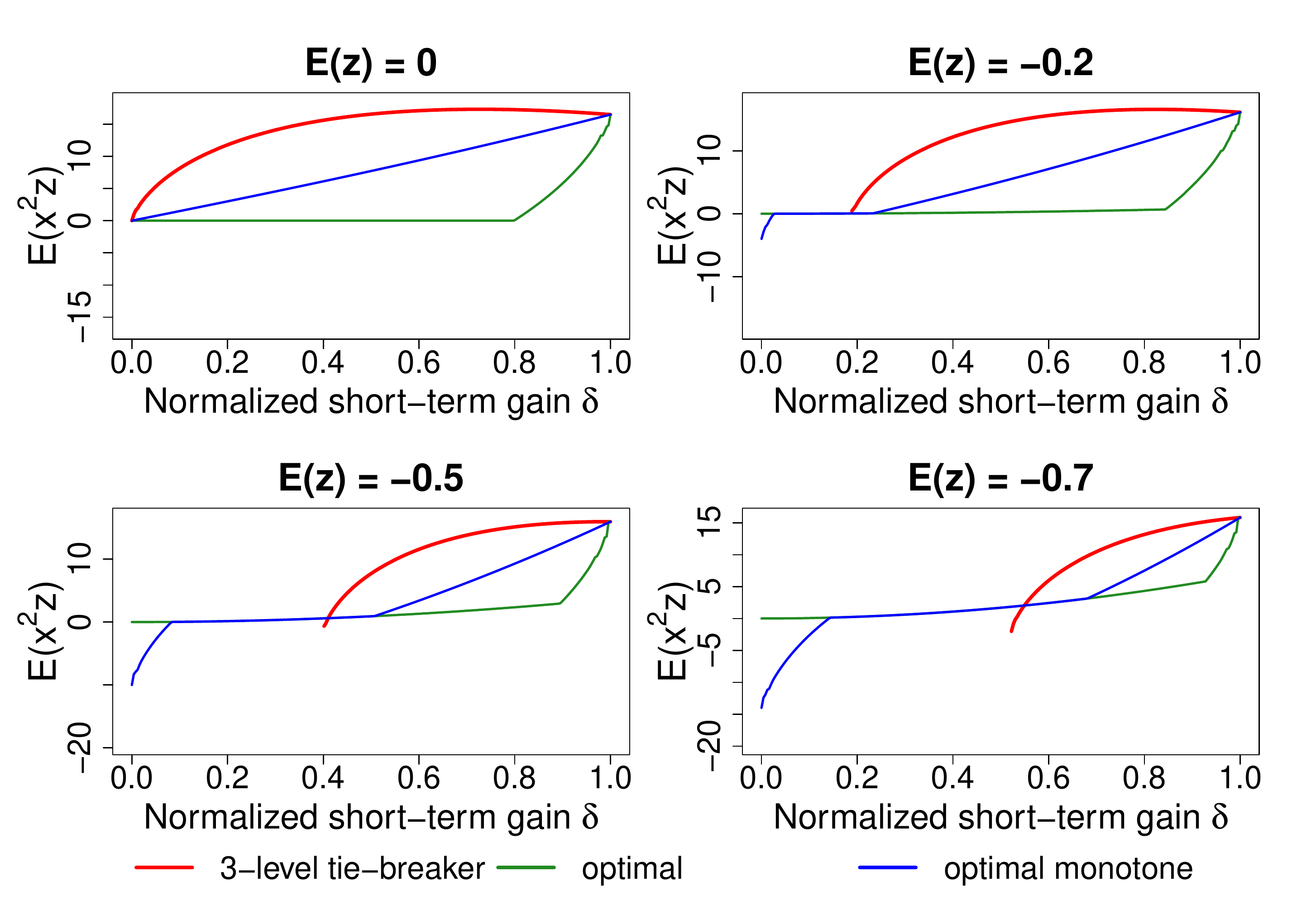} 
\caption{\label{fig:weibull_efficiency}
Same as Figure~\ref{fig:uniform_efficiency}, except for the case where $F$ is a centered Weibull distribution~\eqref{eq:weibull}.}
\end{figure}

We now repeat the analysis of Section~\ref{sec:uniform} for a skewed running variable distribution $F$:
\begin{equation}
\label{eq:weibull}
F(x) = 1-\exp(-\sqrt{x+2})
\end{equation}
for $x \in (-2,\infty) = I$.
This corresponds to a mean-centered Weibull distribution with shape parameter $0.5$ and scale parameter $1$.
Figure~\ref{fig:weibull_efficiency} shows the trade-off curves under this distribution $F$.
We see, as expected by Theorem~\ref{thm:inadmissibility},
that once again the fully randomized design
is inadmissible,
even within the class of monotone designs,
when $\ol{z} \neq 0$.

Another notable feature when $F$ is not symmetric is that the three level tie-breaker is no longer optimal, even in the balanced case $\ol{z}=0$. 
While the unconstrained optimal design 
attains the lower bound $\e(x^2z) = \ol{x^2z}^*(0,\ol{xz};\eff) = 0$ for a wide range of short-term gains, 
Figure~\ref{fig:weibull_efficiency} shows
the three level tie-breaker does not, 
except in the case $\ol{z}=\ol{xz}=0$ corresponding to the RCT. 
In Figure \ref{fig:weibull_efficiency}, 
we see the optimal design is 
over 100 times as efficient as the three level tiebreaker for sufficiently large $\delta$, even in the balanced setting $\ol{z}=0$.
In the unbalanced treatment cases we also see a range of values for which optimal designs with and without the monotonicity constraint
attain the same value of $\e(x^2z)$.
In those situations there exists a globally optimal design that is also monotone.

\subsection{Fixed-$x$ data example}
\label{sec:data_example}
We now illustrate how to compute optimal designs for the original fixed-$x$ problem~\eqref{eq:p_opt_problem_fd} using a real data example.
\citet{ludwigmiller2007} used an RDD to analyze the impact of Head Start,
a U.S.\ government program launched in 1965
that provides benefits such as preschool and health services to children in low-income families.
When the program was launched, extra grant-writing assistance was provided to the 300 counties with the highest poverty rates in the country.
This created a natural discontinuity in the amount of funding to counties as a function of $x$, a county poverty index based on the 1960 U.S. Census.
The distribution of $x$ over $n=2{,}804$ counties is shown in Figure \ref{fig:head_start}.
The data is made freely available by~\citet{cattaneo2017}.

If the government had deemed it ethical to somewhat randomize the 300 counties receiving the grant-writing assistance,
it could have more efficiently estimated the causal impact of this assistance using our $p_{\opt}^{\dag}$, 
while still ensuring poorer counties are preferentially helped,
and no county has a lower chance of getting the assistance
than a more well-off county.
As in the data example of~\citet{klug:owen:2022:tr}, we do not observe the potential outcomes, so we cannot actually implement such a design and compute any estimators.
However, we can still study statistical efficiencies, which depend only on the expected information matrix $\ci$.

\begin{figure}
\includegraphics[width=0.8\linewidth, height=6cm]{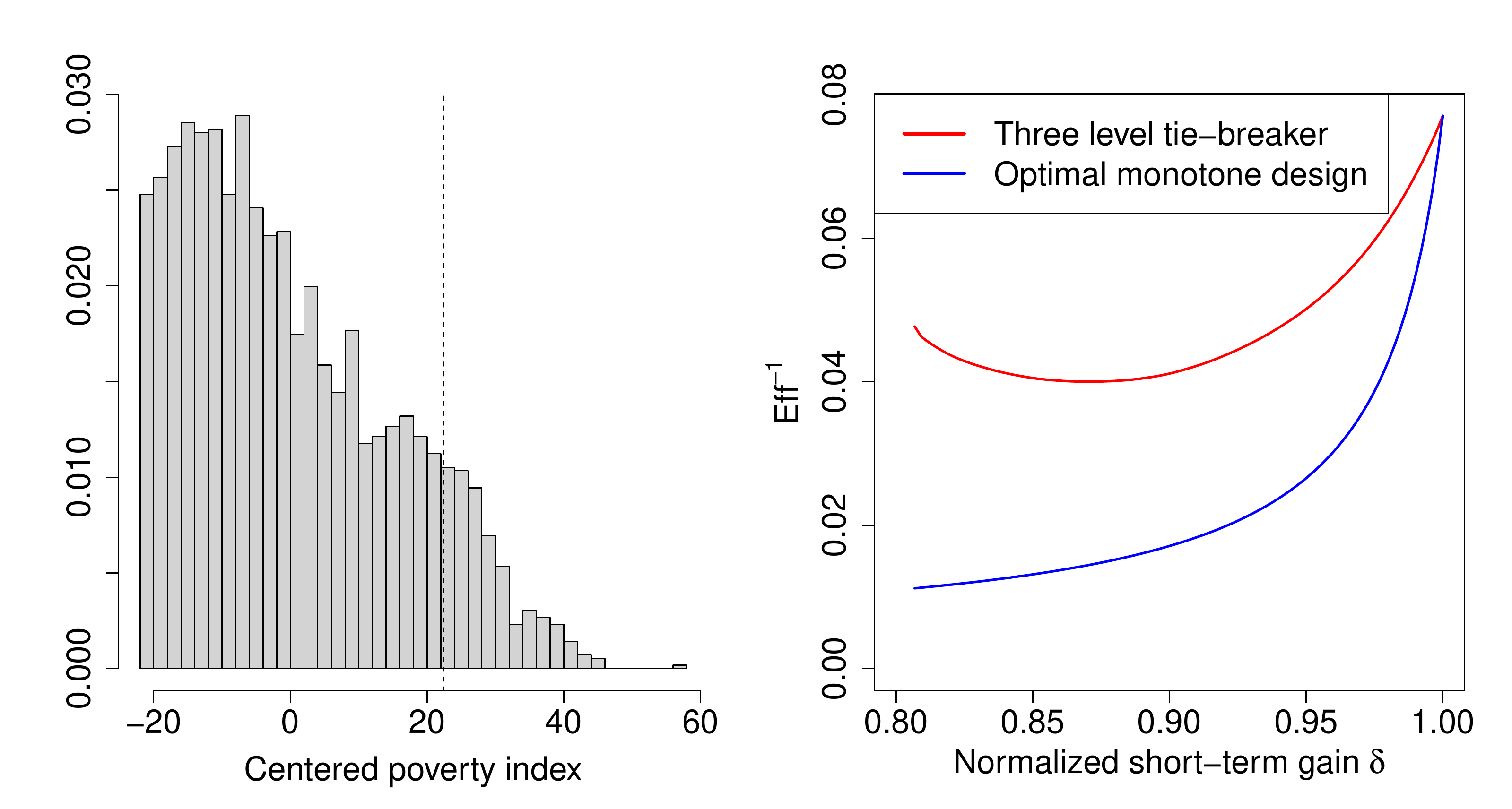} 
\caption{
Left: A histogram of the mean-centered poverty index $x$ for $n=2{,}804$
counties used to determine eligibility for additional grant-writing assistance in the Head Start program. 
The dotted vertical line indicates the eligibility threshold. 
Right: The exploration-exploitation trade-offs for the Head Start data,
comparing the three level tie-breaker with the optimal monotone two level design.
The curves intersect at the value $\eff^{-1}$ of the RDD.}
\label{fig:head_start}
\end{figure}

We fix the treatment fraction at $300/2804$,
corresponding to $\ol{z} \approx -0.79$.
Varying the short-term gain constraint $\ol{xz}$ we seek to compute $p_{\max;\ol{z},\ol{xz}}^{\dag}$ and $p_{\min;\ol{z},\ol{xz}}^{\dag}$.
We describe how to compute the former.
Because $F$ is discrete
it suffices to only consider discontinuity points $t \in \{x_1,\dots,x_n\}$ where $F$ places positive probability mass,
as every design of the form of $p_{\max}^{\dag}$ in~\eqref{eq:gen_p_max_min_mon} has a representation in that form with such $t$.
Also,
given the values of the discontinuity $t$ and $p(t)=\epsilon$,
there is at most one value $\ell=\ell(t,\epsilon) \in [0,1]$
such that the resulting design $p$ in the form of $p_{\max}^{\dag}$ in~\eqref{eq:gen_p_max_min_mon}
satisfies the treatment fraction constraint $\e_p(z)=\ol{z}$.
When such an $\ell$ exists for some $(t,\epsilon)$, call the corresponding design $p^{(t,\epsilon)}$
(note we suppress the dependence on $\ol{z}$).
From Appendix~\ref{app:uniqueness_mon} we deduce that $\e_{p^{(t,\epsilon)}}(xz) < \e_{p^{(t',\epsilon')}}(xz)$ and $\ell(t,\epsilon) > \ell(t',\epsilon')$
if $t > t'$, or if $t=t'$ and $p(t)<p'(t)$.
This shows we can efficiently find the unique $(t,\epsilon)$ so that $p^{(t,\epsilon)}$ satisfies the desired short-term gain constraint
$\e_{p^{(t,\epsilon)}}(xz)=\ol{xz}$.
In particular we compute $t = \max\{s \in \{x_1,\ldots,x_n\} \mid \e_{p^{(s,1)}}(xz) \ge \ol{xz} \}$ via a binary search on $\{x_1,\ldots,x_n\}$,
then solve for $\epsilon$ to satisfy $\e_{p^{(t,\epsilon)}}(xz)=\ol{xz}$.
Given sorted $x$, this entire procedure computes $p_{\max}^{\dag}$ in $O(n)$ operations,
as for each $(t,\epsilon)$, 
$\ell(t,\epsilon)$ and $\e_{p^{(t,\epsilon)}}(xz)$ can be computed in constant time using~\eqref{eq:tower} given the partial sums $\{\sum_{i=1}^m x_i\}_{m=1}^n$.

After computing $p_{\min}^{\dag}$ with a similar approach,
we can apply~\eqref{eq:optimal_design} to compute an optimal design $p_{\opt}^{\dag}$.
As in the continuous case,
we can alternately obtain a solution $p_{\opt}^{\dag}$ of the form in Theorem~\ref{thm:alt_pmax_pmin_mon}
by finding $\ell'$, $u'$, $t'$, and $\epsilon'$ such that
$\e_p(z)=\ol{z}$, $\e_p(xz)=\ol{xz}$, and $\e_p(x^2z)=\ol{x^2z}^*(\ol{z},\ol{xz};\eff)$
for $p(x) \equiv \ell' \indic(x < t') + \epsilon' \indic(x=t') + \indic(x >t')$.
Unlike the continuous case,
we now have 4 unknown parameters instead of 3.
We can search for an optimal set of these parameters by looping through the finite possible values of $t'$ and then doing a univariate search for $\epsilon'$,
noting that knowledge of $t'$ and $\epsilon'$ determines $\ell'$ and $u'$ by the equality constraint parameters $\ol{z},\ol{xz})$.
We implemented this search,
along with the procedure to compute $p_{\max}^{\dag}$ and $p_{\min}^{\dag}$ described above,
in the R language~\citep{R}.
The code is freely available online\footnote{https://github.com/hli90722/optimal\_tiebreaker\_designs}.

The right panel of Figure~\ref{fig:head_start} shows the inverse efficiency for the three level tie-breaker~\eqref{eq:3_level_quantile_tie-breaker}
versus the best two level monotone design obtained by applying the above procedure to the $x_i$ in the Head Start data.
It turns out that for these $x_i$ and our choice of $\ol{z}$, 
$p_{\max}^{\dag}$ is optimal for all $\ol{xz}$ (and hence the unique optimal design, by Proposition~\ref{prop:pmax_pmin_prime}).
We note that with a normalized short term gain $\delta \approx 0.958$, 
which corresponds to random assignment for about 150 counties in the 3-level tie-breaker, 
the optimal monotone two level design has inverse efficiency 0.030, compared to 0.050 for the three level tie-breaker. 
That is, confidence intervals for $\beta_3$ using the three level tie-breaker would be about 29\% wider than for the optimal monotone two-level design,
without additional short-term gain. 
The sharp RDD would give 62\% wider intervals than the optimal monotone two-level design
with only about 4.2\% additional short-term gain.

\section{Summary}
\label{sec:discussion}
Our results provide a thorough characterization of the solutions to a constrained optimal experiment design problem.
Considering a linear regression model for a scalar outcome
involving a binary treatment assignment indicator $z$, a scalar running variable $x$,
and their interaction,
we seek to specify a randomized treatment assignment scheme based on $x$ --- a tie-breaker design ---
that optimizes a statistical efficiency criterion that is an arbitrary continuous function of the expected information matrix under this regression model.
We have equality constraints on the proportion of subjects receiving treatment due to an external budget,
and on the covariance between $x$ and $z$ due to a preference for treating subjects with higher values of $x$.
Critically, our proof techniques,
which deviate from those typically used to show equivalence theorems, 
enable an additional monotonicity constraint.
This allows our results to handle the ethical or economic requirement that a subject cannot have a lower chance of receiving the treatment
than another subject with a lower value of $x$.

In a setting where the running variable $x$ is viewed as random from some distribution $F$ ---
and thus part of the randomness in the expected information matrix defining the efficiency criterion ---
we prove the existence of constrained optimal designs
that stratify $x$ into a small number of intervals
and assign treatment with the same probability to all individuals within each stratum.
In particular, with the monotonicity constraint that is essential in our motivating applications,
we only need three strata,
one of which only contains a single running variable value.
We also provide strong conditions on which the optimal tie-breaker design is unique.
We emphasize the generality of our results,
which apply for any continuous efficiency criterion,
any running variable distribution $F$ (subject only to weak moment existence conditions),
and the full range of feasible equality constraints.
The problem an investigator faces in practice,
where there are a finite number of running variable values $x_1,\ldots,x_n$
known (hence non-random) at the time of treatment assignment,
is a special case of our more general problem where $F$ is discrete and takes on values $x_1,\ldots,x_n$ with equal probability.
This enables optimal designs to be easily computed in practice, as described in Section~\ref{sec:data_example}.

We believe that this work provides a useful starting point to study optimal tie-breaker designs.
For results on tie-breaker designs beyond the two line
parametric regression,
see \cite{mvtiebreaker}
for a multivariate regression context
and \cite{klug:owen:2022:tr} for local linear
regression models with a scalar running variable.
\begin{appendix}
\section{Proof of Corollary \ref{corollary:M11}} \label{app:M11}
For any $\ol{z} \in (-1,1)$ we have $\ol{xz}_{\max}(\ol{z}) = \e_{p_{\ol{z}}}(xz) = 2\e(xp_{\ol{z}}(x))$ where $p_{\ol{z}}$ is as in Lemma~\ref{lemma:stg}. The desired condition $M_{11} > 0$ is equivalent to $(\ol{xz})^2 < \e(x^2)(1-\ol{z}^2)$
and so it suffices to show $\e(xp_{\ol{z}}(x))^2 < \e(x^2)((1+\ol{z})/2)((1-\ol{z})/2)$.

Applying Cauchy-Schwarz
to $xp_{\ol{z}}(x)^{1/2} \times p_{\ol{z}}(x)^{1/2}$ 
and then $x(1-p_{\ol{z}}(x))^{1/2} \times (1-p_{\ol{z}}(x))^{1/2}$ yields the two equations
\begin{align*}
\e(xp_{\ol{z}}(x))^2 &<  \e(x^2p_{\ol{z}}(x))
\Bigl(\frac{1+\ol{z}}{2}\Bigr) \\
\e(xp_{\ol{z}}(x))^2 = \e(x(1-p_{\ol{z}}(x)))^2 &< \e(x^2(1-p_{\ol{z}}(x)))\Bigl(\frac{1-\ol{z}}{2}\Bigr)
\end{align*}
where we have used the fact $\e(xp_{\ol{z}}(x)) = -\e(x(1-p_{\ol{z}}(x)))$ since $\e(x)=0$.
Note both inequalities are strict, since $x$ cannot equal a scalar multiple of $p_{\ol{z}}(x)$ w.p.1. 
If it did, then $\e(kp_{\ol{z}}(x)-x)=k(1+\ol{z})/2 = 0$ for some $k$, implying $k=0$ and hence $x=0$ w.p.1, contradicting $\var(x)>0$.
As $\e(x^2) = \e(x^2p_{\ol{z}}(x)) + \e(x^2(1-p_{\ol{z}}(x)))$ we know that either $\e(x^2p_{\ol{z}}(x)) \leq \e(x^2) \cdot (1+\ol{z})/2$ or $\e(x^2(1-p_{\ol{z}}(x))) \leq \e(x^2) \cdot (1-\ol{z})/2$.
\quad$\Box$

\section{Proof of uniqueness in Proposition~\ref{prop:pmax_pmin}}
\label{app:uniqueness}
We show uniqueness for $p_{\min}$.
The same argument shows uniqueness for $1-p_{\max}$ and hence uniqueness for $p_{\max}$.  
Suppose that $p(x)=\delta_1\indic(x=b_1)+\indic(b_1<x<b_2)+\delta_2\indic(x=b_2)$ and $p'(x)=\delta'_1\indic(x=b'_1)+\indic(b'_1<x<b_2')+\delta'_2\indic(x=b'_2)$ are both solutions for $p_{\min}$. 
By symmetry we can assume that either $b_1 < b_1'$, or both $b_1=b_1'$ and $\delta_1 \le \delta_1'$.
Since $p$ and $p'$ are feasible for~\eqref{eq:max_min_inequality}, we must have $\e(p(x)-p'(x))=0$ and $\e(x(p(x)-p'(x)))=0$, in view of~\eqref{eq:tower}.
We show that $p(x)=p'(x)$ w.p.1. under $x \sim F$.
Note that we can assume without loss of generality that $\Pr( x\in[b_1,b_1+\epsilon))>0$ for any $\epsilon>0$,
because otherwise, we could increase $b_1$ to $b_1 + \sup\{\epsilon > 0 \mid \Pr(x \in [b_1,b_1+\epsilon))=0\}$
without changing $p$ on a set of positive probability. 
We can similarly assume that $\Pr(x\in(b_2-\epsilon,b_2])>0$ for any $\epsilon>0$.
Finally, we impose these two canonicalizing conditions on $b_1'$ and $b_2'$ as well. 

Assume first that $b_2>b_1$.
Then we cannot have $b_1'>b_1$ because we would then need either $b_2'>b_2$ or $b_2'=b_2$ with $\delta_2'>\delta_2$ and $\Pr(x=b_2)>0$ to enforce $\e(p(x)-p'(x))=0$ and this would cause $\e(x(p(x)-p'(x)))<0$. 
We similarly cannot have $b_1'=b_1$ with both $\delta_1'>\delta_1$ and $\Pr(x=b_1)>0$.
Therefore after canonicalizing, we know that both $p$ and $p'$ are equivalent to designs of the form given with $b_1=b_1'$ and $\delta_1=\delta_1'$ 
along with the analogous conditions $b_2=b_2'$ and $\delta_2=\delta_2'$. 
Then our canonicalized $p$ and $p'$ satisfy $p(x)=p'(x)$ for all $x$ and so in particular $\Pr(p(x)=p'(x))=1$. 

It remains to handle the case where $b_1=b_2$. We then have 
$\Pr( x=b_1)>0$ since $\ol{z}>-1$. 
If $b_1'>b_1$ then the support of $p'$ is completely to the right of that of $p$ which violates $\e(x(p(x)-p'(x)))=0$. We can similarly rule out $b_2'<b_2$. As a result $p'$ must have $b_1'\le b_1=b_2\le b_2'$. Then we must have $\delta_1'\Pr(x=b_1')+\Pr( b_1'<x<b_1)=0$ or else $\e(p'(x)-p(x))>0$. For the same reason, we must have $\Pr(b_2<x<b_2')+\delta_2'\Pr(x=b_2')=0$. It then follows that both $p$ and $p'$ have support $\{b_1\}$ and then $\e(p(x))=\e(p'(x))=(1+\ol{z})/2$ forces $(1+\ol{z})/(2\Pr(x=b_1)) = \delta_1=p(b_1)=p'(b_1)=\delta_1'$,
so $\Pr( p(x)=p'(x))=1$.

\section{Proof of uniqueness in Proposition~\ref{prop:pmax_pmin_prime}}
\label{app:uniqueness_mon}
We focus on $p_{\max}^{\dag}$ and
consider two monotone designs $p$ and $p'$ satisfying the feasibility constraints $\e_p(z)=\e_{p'}(z)=\ol{z}$ and $\e_p(xz)=\e_{p'}(xz)=\ol{xz}$
along with the characterization of $p_{\max}^{\dag}$ in~\eqref{eq:gen_p_max_min_mon}.
Then $p(x)=\ell\indic(x<t)+\delta\indic(x=t)+\indic(x>t)$ and $p'(x)=\ell'\indic(x<t')+\delta'\indic(x=t')+\indic(x>t')$
for some $\ell, \ell' \in (0,1)$ with $\ell \le \delta \le 1$ and $\ell'\le\delta' \le 1$.
Note the cases $\ell \in \{0,1\}$ (and the same for $\ell'$) are excluded by the assumptions that $\ol{xz} < \ol{xz}_{\max}(\ol{z})$ and $\ol{z}<1$.
Also, $\ol{z}<1$ also guarantees $\min(\Pr(x\le t), \Pr(x \le t'))>0$.
Finally, we note that we only have to show $p(x)=p'(x)$ for almost all $x \neq t$,
since then $\e(p(x)-p'(x))=0$ ensures either $p(t)=p'(t)$ or $\Pr(x=t)=0$;
in either case this gives $p=p'$ w.p.1. 
By symmetry we can assume that $t \le t'$ with $\delta \equiv p(t) \ge p(t') =: \delta'$ if $t=t'$. 
Then $p(x)=p'(x)$ for all $x > t'$.

Now we compute
\begin{align*}
\e(x(p(x)-p'(x))) & = \e((x-t)(p(x)-p'(x))) \text{ since $\e(p(x))=\e(p'(x))$} \\
& = \e((t-x)(p'(x)-p(x))\indic(x < t)) + \e((x-t)(p(x)-p'(x))\indic(t < x \le t')) \\
& = (\ell'-\ell)\e((t-x)\indic(x<t)) + \e((x-t)(1-p'(x))\indic(t < x \le t')).
\end{align*}

If $t=t'$, then the right-hand side reduces to just $(\ell'-\ell)\e((t-x)\indic(x<t))$.
This is nonzero unless $\Pr(x<t)=0$ or $\ell=\ell'$.
In both cases $p(x)=p'(x)$ for almost all $x \neq t$.

If $t < t'$, then we can assume $\Pr(t<x\le t') > 0$ (otherwise the problem reduces to the case $t=t'$).
First suppose $\ell \ge \ell'$. Then $p(x) \ge p'(x)$ for all $x$ and so the treatment fraction constraint would require the identity
\[
\indic(t < x \le t') = p(x)\indic(t<x \le t') \ge p'(x)\indic(t<x \le t') = \delta'\indic(x=t') + \ell'\indic(t<x<t')
\]
to hold with equality w.p.1.
But since $\ell'<1$, equality w.p.1. can only occur if $\Pr(t<x<t')=0$ and $\delta'=1$.
In that case we immediately see $p(x)=p'(x)$ for almost all $x >t$,
but $0=\e(x(p(x)-p'(x))) = (\ell'-\ell)\e((t-x)\indic(x<t))$ so $p(x)=p'(x)$ for almost all $x<t$ as well.
Conversely, if we suppose $\ell < \ell'$, then $\e(x(p(x)-p'(x)))=0$ requires $\Pr(x<t)=0$ and $p'(x)=1$ for almost all $x \in (t,t']$,
so once again $p(x)=p'(x)$ for almost all $x \neq t$.


\section{Proof of Theorem~\ref{thm:opt_eff_monotonicity}}
\label{app:opt_eff_monotonicity}
For any feasible design $p$,
we have
\begin{align}
\label{eq:detM}
\begin{split}
\det(M)
&= -\frac{(1-\ol{z}^2) \cdot \bigl(\e_p(x^2z)\bigr)^2}{\e(x^2)}  - \frac{2\ol{z}(\ol{xz})^2 \cdot \e_p(x^2z)}{\e(x^2)}  
+ \e(x^2)\bigl(1-\ol{z}^2\bigr)+(\ol{xz})^2\biggl(\frac{(\ol{xz})^2}{\e(x^2)}-2\biggr).
\end{split}
\end{align}
where $M=M(p)$ as in~\eqref{eq:M}.
Thus $\det(M(p))$ is a concave quadratic function of $\e_p(x^2z)$
globally maximized at 
\begin{equation}
\label{eq:a_star}
a^*(\ol{z},\ol{xz}) \equiv -\frac{\ol{z} \cdot (\ol{xz})^2}{1-\ol{z}^2}
\end{equation}
It follows that $\ol{x^2z}^*(\ol{z},\ol{xz};\eff)$
is the point in $I_{\cf}(\ol{z},\ol{xz}) = [I_{\cf;\min}(\ol{z},\ol{xz}), I_{\cf;\max}(\ol{z},\ol{xz})]$
closest in absolute value to $a^*(\ol{z},\ol{xz})$,
i.e.
\begin{equation}
\label{eq:quad_min}
\ol{x^2z}^*(\ol{z},\ol{xz};\eff) = 
\begin{cases}
I_{\cf;\min}(\ol{z},\ol{xz}) & a^*(\ol{z},\ol{xz}) \le I_{\cf;\min}(\ol{z},\ol{xz}) \\
a^*(\ol{z},\ol{xz}) & I_{\cf;\min}(\ol{z},\ol{xz}) < a^*(\ol{z},\ol{xz}) < I_{\cf;\max}(\ol{z},\ol{xz}) \\
I_{\cf;\max}(\ol{z},\ol{xz}) & a^*(\ol{z},\ol{xz}) \ge I_{\cf;\max}(\ol{z},\ol{xz})
\end{cases}
\end{equation}
The above holds for any choice of $\cf$;
for the remainder of this proof we take $\cf$
to be the set of all measurable design functions.

We first show the case where $\ol{z}=0$.
Note that $\e_p(x^2z)=0=a^*(0,\ol{xz})$
for any symmetric design $p$,
by symmetry of the running variable distribution.
By continuity,
for any $\ol{xz} \in [0,\ol{xz}_{\max}(\ol{z})]$
there exists $\Delta \in [0,\infty]$
such that the three level tie-breaker $p_{3;\ol{z},\Delta}$ 
(which is symmetric and always satisfies $\e_{p_{3;\ol{z},\Delta}}(z)=0$)
satisfies $\e_{p_{3;\ol{z},\Delta}}(xz)=\ol{xz}$ too.
This shows that for all $\ol{xz} \in [0,\ol{xz}_{\max}(\ol{z})]$,
$0 \in I_{\cf}(\ol{z},\ol{xz})$
and hence $\ol{x^2z}^*(0,\ol{xz};\eff)=0$,
meaning any feasible design $p$ with $\e_p(x^2z)=0$ is optimal.
Then by~\eqref{eq:detM}
\[
\det(M(p_{\opt;0,\ol{xz}})) = \e(x^2) + (\ol{xz})^2\biggl(\frac{(\ol{xz})^2}{\e(x^2)}-2\biggr)
\]
which is decreasing in $\ol{xz}$ on $[0,\ol{xz}_{\max}(\ol{z})]$,
showing the theorem for $\ol{z}=0$.

For the cases $\ol{z} < 0$ and $\ol{z}>0$, we begin with the two following claims.
\begin{claim} \label{claim:p_opt_min}
For any $(\ol{z},\ol{xz}) \in \cj$ with $\ol{z} < 0$, we have $I_{\cf;\min}(\ol{z},\ol{xz}) \leq 0 < a^*(\ol{z},\ol{xz})$.
\end{claim}
\begin{claim} \label{claim:p_opt_max}
For any $(\ol{z},\ol{xz}) \in \cj$ with $\ol{z} > 0$, we have $I_{\cf;\max}(\ol{z},\ol{xz}) \geq 0 > a^*(\ol{z},\ol{xz})$.
\end{claim}
\begin{proof}[Proof of Claims~\ref{claim:p_opt_min} and~\ref{claim:p_opt_max}]
We write $I_{\cf;\min}(\ol{z},\ol{xz}) = 2\e(x^2p_{\min;\ol{z},\ol{xz}}(x))-\e(x^2)$
and similarly rewrite $I_{\cf;\max}(\ol{z},\ol{xz})$.
For claim~\ref{claim:p_opt_min}, we proceed by writing $p_{\min}=p_{[b_1,b_2]}$ by Proposition~\ref{prop:pmax_pmin}
(suppressing the dependence of $b_1$ and $b_2$ on $(\ol{z},\ol{xz})$ in our notation)
and performing casework on the signs of $b_1$ and $b_2$ to show that $\e(x^2p_{\min}(x)) \leq \e(x^2)/2$ in each case. 
In the case $b_1 \leq b_2 \leq 0$ we have $\e(x^2p_{\min}(x)) \leq \e(x^2\bm{1}(x \leq 0)) = \e(x^2)/2$ by symmetry; 
similarly if $b_2 \ge b_1 \ge 0$ then $\e(x^2p_{\min}(x)) \leq \e(x^2\bm{1}(x \geq 0)) = \e(x^2)/2$. 
Next, if $b_1 \leq 0 \leq -b_1 \leq b_2$ then $F(b_2)+F(0)-F(b_1) = \Pr(b_1 < x \le b_2)+1/2 = \e(p_{\min}(x))+1/2 < 1$ since $\ol{z}<0$ implies $\e(p(x))<1/2$ by~\eqref{eq:tower}. 
Therefore
\begin{align*}
\e(x^2p_{\min}(x)) & = \e(x^2\indic(b_1 \leq x \leq 0)) + \e(x^2\indic(0 \leq x \leq b_2))  \\
& \leq b_2^2 \Pr(b_1 \leq x \leq 0) + \e(x^2\indic(0 \leq x \leq b_2)) \\
& = b_2^2 \Pr(b_2 \leq x \leq F^{-1}(F(b_2)+F(0)-F(b_1)))+ \e(x^2\indic(0 \leq x \leq b_2)) \\
& \leq \e(x^2\bm{1}(0 \leq x \leq F^{-1}(F(b_2)+F(0)-F(b_1)))) \leq \frac{\e(x^2)}{2}
\end{align*}
where the final inequality uses symmetry of $F$ again.
The final case $b_1 \leq 0 \leq b_2 \leq -b_1$ follows by a symmetric argument.
The proof of Claim~\ref{claim:p_opt_max} is completely analogous, 
with $p_{\max}=p_{[a_1,a_2]^c}$ by Proposition~\ref{prop:pmax_pmin}.
\end{proof}

We now proceed to prove the theorem.
Given Claim~\ref{claim:p_opt_min}, we have $\ol{x^2z}^*(\ol{z},\ol{xz};\eff) = \min(I_{\cf;\max}(\ol{z},\ol{xz}), a^*(\ol{z},\ol{xz}))$ by~\eqref{eq:quad_min}, 
and hence suppressing some $\ol{z}$ dependences
\[
h(\ol{xz}) \equiv \det(M(p_{\opt;\ol{z},\ol{xz}})) = 
\begin{cases}
h^*(\ol{xz}), & g(\ol{xz}) \geq a^*(\ol{z},\ol{xz})  \\
\det(M(p_{\max;\ol{z},\ol{xz}})), & g(\ol{xz}) \leq a^*(\ol{z},\ol{xz})
\end{cases}
\]
where $g(\ol{xz}) \equiv I_{\cf;\max}(\ol{z},\ol{xz})$ 
and $h^*(\ol{xz})$ is defined by substituting $\e_p(x^2z)=a^*(\ol{z},\ol{xz})$ into~\eqref{eq:detM}. 
We must show that $h(\ol{xz})$ is decreasing on $\ol{xz}>0$.

First, we compute
$h^*(\ol{xz}) = (\e(x^2))^2 M_{11}^2 /(\e(x^2)(1-\ol{z}^2))$
and note it is decreasing in $\ol{xz}$ since $M_{11}$ is positive (Corollary~\ref{corollary:M11})
and decreasing in $\ol{xz}$ on $[0,\ol{xz}_{\max}(\ol{z})]$.
Next, we show $\det(M(p_{\max;\ol{z},\ol{xz}}))$ is decreasing in $\ol{xz}$.
Note $(a_1,a_2)$ are the unique solutions to the system
\begin{align*}
F(a_1) + 1-F(a_2) & = (1+\ol{z})/2  \\[1ex]
\e(x(\ind(x < a_1) +\ind(x > a_2))) & = \ol{xz}/2
\end{align*}
By the implicit function theorem (e.g., \citet{Oliveira2018} 
since we do not require continuity of $f$)),
it follows that $a_1=a_1(\ol{xz})$ and $a_2=a_2(\ol{xz})$ are differentiable and satisfy
\begin{align}
f(a_1) a_1'(\ol{xz}) - f(a_2)a_2'(\ol{xz}) & = 0  \label{eq:partial_a_xz_1},\quad\text{and} \\
a_1 f(a_1) a_1'(\ol{xz}) - a_2f(a_2)a_2'(\ol{xz}) & = 1/2 \label{eq:partial_a_xz_2}.
\end{align} 
Equations~\eqref{eq:partial_a_xz_1} and~\eqref{eq:partial_a_xz_2} imply that
\[
g'(\ol{xz}) = \frac{\partial}{\partial \ol{xz}} \e_{p_{[a_1,a_2]^c}}(x^2z) = 2a_1^2 f(a_1) a_1'(\ol{xz}) - 2a_2^2f(a_2)a_2'(\ol{xz}) = a_1+a_2 < 0.
\]
The inequality follows by the assumption $\ol{z}<0$ which ensures $a_1$ and $a_2$ must have different signs,
and then noting that $\e_{p_{\max}}(xz)>0$ requires $\e(xp_{\max}(x)\indic(x>0)) > -\e(xp_{\max}(x)\indic(x<0)) = \e(xp_{\max}(-x)\indic(x>0))$,
the equality following by symmetry of $F$.
Thus, for all $\ol{xz}>0$ such that $g(\ol{xz}) \leq a^*(\ol{z},\ol{xz}) = -\ol{z} (\ol{xz})^2/(1-\ol{z}^2)$ we have
\begin{align*}
\frac{\e(x^2) \partial \det(M(p_{\max}))}{\partial \ol{xz}} 
& = g(\ol{xz})\left(-2(1-\ol{z}^2) g'(\ol{xz}) - 4(\ol{xz} \cdot \ol{z})\right) \\
&\phantom{=}\, - 2(\ol{xz})^2 \cdot \ol{z} g'(\ol{xz}) + 4(\ol{xz})((\ol{xz})^2-\e(x^2)) \\
& \leq \frac{4(\ol{xz})^3(\ol{z})^2}{1-\ol{z}^2} + 4(\ol{xz})((\ol{xz})^2-\e(x^2))  = -\frac{4M_{11} \cdot (\ol{xz})\e(x^2)}{1-\ol{z}^2}
\end{align*}
The RHS is negative (Corollary~\ref{corollary:M11}),
so $\det(M(p_{\max;\ol{z},\ol{xz}}))$ is in fact decreasing in $\ol{xz} > 0$. 

Finally, we fix $0 \leq x_1 < x_2 \leq \ol{xz}_{\max}(\ol{z})$ and show $h(x_1)>h(x_2)$.
Note $\ovl{g}(\ol{xz}) \equiv g(\ol{xz})-a^*(\ol{z},\ol{xz})$ is continuous in $\ol{xz}$, 
and $h^*(\ol{xz}) \geq \det(M(p_{\max;\ol{z},\ol{xz}}))$. We now carry out casework on the signs of $\bar{g}(x_1)$ and $\bar{g}(x_2)$. \\
\underline{$\ovl{g}(x_1) \geq 0$:} In this case
\begin{equation}
\label{eq:h_case1}
h(x_1) = h^*(x_1) > h^*(x_2) \geq h(x_2).
\end{equation}
\underline{$\ovl{g}(x_1) < 0$ and $\ovl{g}(x_2) \geq 0$:} Define $S = \{x \in [x_1,x_2] \mid \ovl{g}(x) \geq 0\}$, which contains $x_2$.
Letting $x_3 = \inf S > x_1$, we have $\ovl{g}(x_3)=0$ and $\ovl{g}(x) \leq 0 $ for $x \in [x_1,x_3]$, so 
\begin{equation}
\label{eq:h_case2}
h(x_1) = \det(M(p_{\max;\ol{z},x_1})) > \det(M(p_{\max;\ol{z},x_3})) = h^*(x_3) \geq h^*(x_2)=h(x_2)
\end{equation}
\underline{$\ovl{g}(x_1) < 0$ and $\ovl{g}(x_2) < 0$:} In this case either $\ovl{g}(x) \leq 0$ on $[x_1,x_2]$ (so $h(x_1) = \det(M(p_{\max;\ol{z},x_1})) > \det(M(p_{\max;\ol{z},x_2})) = h(x_2)$),
or $S$ as defined in the previous case is non-empty with $x_3 = \inf(S)$ and $x_4 = \sup(S)$ satisfying $x_1 < x_3 \leq x_4 < x_2$ and $\ovl{g}(x_3)=\ovl{g}(x_4)=0$. Then 
\[
h(x_1) \stackrel{(\ref{eq:h_case2})}{>} h(x_3) \stackrel{(\ref{eq:h_case1})}{\geq} h(x_4) \stackrel{(\ref{eq:h_case2})}{>} h(x_2)
\]
which shows the theorem when $\ol{z}<0$. The proof of the case $\ol{z} > 0$ is completely symmetric, and relies on Claim 2.


\section{Proof of Theorem~\ref{thm:inadmissibility}}
\label{app:inadmissibility}
First we fix $\ol{z}<0$. 
It suffices to show that assuming $\e(x^2) < F^{-1}(1)^2$,
there exists $\delta > 0$ such that
$(\partial/\partial \ol{xz}) \det\bigl(M( p_{\opt;\ol{z},\ol{xz}}^{\dag})\bigr) > 0$
whenever $\ol{xz} \in (0,\delta)$, 
and that $\det(M(p_{\opt;\ol{z},\ol{xz}}^{\dag}))$ is continuous in $\ol{xz}$ at $\ol{xz}=0$. 

From the assumed continuity of $F$ and Proposition~\ref{prop:pmax_pmin_prime},
we have $p_{\max}^\dag(x) = p_{\ell,1,t}(x)$, with $\ol{z}>-1$ ensuring $F(t) > 0$.
Again, we suppress the dependence of $\ell$ and $t$ on $(\ol{z},\ol{xz})$ in our notation for brevity.
By the treatment fraction constraint $\e_{p_{\ell,1,t}}(z)=\ol{z}$, we must have $\ell=1-(1-\ol{z})/(2F(t))$.
From the short-term gain constraint $\e_{p_{\ell,1,t}}(xz)=\ol{xz}$ we see
\begin{align*}
\frac{\ol{xz}}{2} & =(\ell-1)\e(x\indic(x<t))= -\frac{1-\ol{z}}{2F(t)}\e(x\indic(x<t)).
\end{align*}
We know by Proposition~\ref{prop:pmax_pmin_prime} and continuity of $F$ 
that the two equations above have a unique solution $(\ell,t)=(\ell(\ol{xz}),t(\ol{xz}))$ 
for $\ol{xz} \in (0,\ol{xz}_{\max}(\ol{z}))$.
Thus, we can differentiate both of the equations above with respect to $\ol{xz}$ to see that the derivatives of $\ell$ and $t$ are given by
\[
t' = t'(\ol{xz}) = \frac{F(t)^2}{(1-\ol{z})f(t)\e((x-t)\ind(x<t))}\qquad\text{and} \qquad 
\ell'  = \ell'(\ol{xz}) = \frac{1}{2\e((x-t)\ind(x < t))}.
\]
Then $g(\ol{xz}) \equiv I_{\cf;\max}(\ol{z},\ol{xz}) = 2(\ell-1)\e(x^2\ind(x < t)) + \e(x^2)$
is differentiable as well with
\[
g'(\ol{xz}) = 
2\ell'\e(x^2\indic(x<t))
+2(\ell-1)t^2f(t)t' =
\frac{2(1-\ell)t^2F(t)^2-(1-\ol{z})\e(x^2\ind(x < t))}{(1-\ol{z})\e((t-x)\ind(x<t))}
\]
Next, note that
$g(0) = \ol{z}\e(x^2) < 0 = a^*(\ol{z},0)$,
in the notation of~\eqref{eq:a_star}.
By differentiability (and thus continuity) of $a^*(\ol{z},\cdot)$ and $g$ (the latter due to differentiability of $\ell$ and $t$), 
we conclude that there exists $\epsilon > 0$ such that $a^*(\ol{z},\ol{xz}) - g(\ol{xz}) \geq 0$ for all $\ol{xz} \in [0,\epsilon]$. 
By~\eqref{eq:quad_min},
this means $p_{\opt;\ol{z},\ol{xz}}^{\dag} = p_{\max,\ol{z},\ol{xz}}^{\dag}$ for all $\ol{xz} \in [0,\epsilon]$.
Thus, it suffices to show 
$\frac{\partial}{\partial \ol{xz}} \det(M(p_{\max;\ol{z},\ol{xz}}^{\dag})) > 0$ 
for all $\ol{xz} \in (0,\delta)$, 
for some $\delta \le \epsilon$. 
Continuity of $\det(M(p_{\max;\ol{z},\ol{xz}}^{\dag}))$ at $\ol{xz}=0$ 
follows immediately from continuity of $g$ and~\eqref{eq:detM}. 

As $\ol{xz}\downarrow 0$ we have $t(\ol{xz})\uparrow F^{-1}(1)$ and $\ell(\ol{xz})\uparrow (1+\ol{z})/2$ and also $\e((t-x)\ind(x<t)) = tF(t)-\e(x\ind(x<t)) \uparrow F^{-1}(1)$.
In the case $F^{-1}(1) < \infty$ we have $\lim_{\ol{xz} \downarrow 0} g'(\ol{xz}) = F^{-1}(1) - \e(x^2)/F^{-1}(1) > 0$ by assumption.
If $F^{-1}(1) = \infty$, then $g'(\ol{xz}) \rightarrow \infty$ as $\ol{xz} \downarrow 0$. 
Finally, we substitute into the formula~\eqref{eq:detM} for $\det(M)$ getting
\begin{align*}
\frac{\partial \det(M( p_{\max;\ol{z},\ol{xz}}^{\dag}))}{\partial \ol{xz}}  &= -\frac{2g'(\ol{xz})\left((1-\ol{z}^2)g(\ol{xz})+\ol{z}(\ol{xz})^2\right)}{\e(x^2)} - \frac{4g(\ol{xz})\ol{z}(\ol{xz})}{\e(x^2)} + \frac{4(\ol{xz})^3}{\e(x^2)}-4\ol{xz}.
\end{align*}
Since $g(\ol{xz}) \to  \ol{z} \cdot \e(x^2)$ as $\ol{xz} \downarrow 0$, we have
$(1-\ol{z}^2)g(\ol{xz})+(\ol{xz})^2\ol{z} \to \e(x^2)\ol{z}M_{11} < 0$ (Corollary~\ref{corollary:M11}).
Our analysis of the limiting behavior on $g'(\ol{xz})$ then indicates that
\[
\lim_{\ol{xz} \downarrow 0} \frac{\partial \det(M(p_{\max;\ol{z},\ol{xz}}^{\dag}))}{\partial \ol{xz}}  = -2\ol{z}M_{11}(F^{-1}(1)-\e(x^2)/F^{-1}(1)) > 0
\]

The proof for the case $\ol{z} > 0$ is completely analogous. 
We first show that $p_{\opt;\ol{z},\ol{xz}}^{\dag} = p_{\min;\ol{z},\ol{xz}}^{\dag}$ whenever $\ol{xz}$ is sufficiently close to 0. 
Then we note $(u,s)$ is the unique solution to the equations $u=(1+\ol{z})/(2(1-F(s)))$ and $\ol{xz}/2 = (1+\ol{z})\e(x\indic(x \ge s))/(2(1-F(s)))$
to compute the derivatives $u'(\ol{xz})$ and $s'(\ol{xz})$. 
This enables us to show 
\[
\lim_{\ol{xz} \downarrow 0} (\partial/\partial \ol{xz}) \det(M( p_{\min;\ol{z},\ol{xz}}^{\dag})) > 0
\]
under the condition $\e(x^2) < F^{-1}(0)^2$.

\end{appendix}

\section*{Acknowledgments}
This work was supported by the US National Science Foundation
under grants IIS-1837931 and DMS-2152780.
The authors would like to thank Kevin Guo, Dan Kluger, Tim Morrison, and several anonymous reviewers for helpful comments.

\bibliographystyle{apalike}
\bibliography{rd}

\begin{thebibliography}{}

\bibitem[Abdulkadiroglu et~al., 2017]{abdulkadiroglu2017impact}
Abdulkadiroglu, A., Angrist, J.~D., Narita, Y., and Pathak, P.~A. (2017).
\newblock Impact evaluation in matching markets with general tie-breaking.
\newblock Technical report, National Bureau of Economic Research.

\bibitem[Aiken et~al., 1998]{aike:west:schw:carr:hsiu:1998}
Aiken, L.~S., West, S.~G., Schwalm, D.~E., Carroll, J.~L., and Hsiung, S.
  (1998).
\newblock Comparison of a randomized and two quasi-experimental designs in a
  single outcome evaluation: Efficacy of a university-level remedial writing
  program.
\newblock {\em Evaluation Review}, 22(2):207--244.

\bibitem[Angrist et~al., 2020]{angrist2020}
Angrist, J., Autor, D., and Pallais, A. (2020).
\newblock Marginal effects of merit aid for low-income students.
\newblock Technical report, National Bureau of Economic Research.

\bibitem[Atkinson et~al., 2007]{atkinsonetal2007}
Atkinson, A., Donev, A., and Tobias, R. (2007).
\newblock {\em Optimum experimental designs, with {SAS}}.
\newblock Oxford University Press.

\bibitem[Atkinson, 1982]{atkinson1982}
Atkinson, A.~C. (1982).
\newblock Optimum biased coin designs for sequential clinical trials with
  prognostic factors.
\newblock {\em Biometrika}, 69(1):61--67.

\bibitem[Atkinson, 2014]{atkinson2014}
Atkinson, A.~C. (2014).
\newblock Selecting a biased-coin design.
\newblock {\em Statistical Science}, 29(1):144--163.

\bibitem[Bandyopadhyay and Biswas, 2001]{bandyopadhyaybiswas2001}
Bandyopadhyay, U. and Biswas, A. (2001).
\newblock Adaptive designs for normal responses with prognostic factors.
\newblock {\em Biometrika}, 88(2):409--419.

\bibitem[Biswas and Bhattacharya, 2018]{biswasbhattacharya2018}
Biswas, A. and Bhattacharya, R. (2018).
\newblock A class of covariate-adjusted response-adaptive allocation designs
  for multitreatment binary response trials.
\newblock {\em Journal of biopharmaceutical statistics}, 28(5):809--823.

\bibitem[Boyd and Vandenberghe, 2004]{boyd:vand:2004}
Boyd, S. and Vandenberghe, L. (2004).
\newblock {\em Convex optimization}.
\newblock Cambridge University Press, Cambridge.

\bibitem[Brent, 1973]{brent1973}
Brent, R.~P. (1973).
\newblock {\em Algorithms for minimization without derivatives}.
\newblock Prentice-Hall, Inc., Englewood Cliffs, NJ.

\bibitem[Campbell, 1969]{camp:1969}
Campbell, D.~T. (1969).
\newblock Reforms as experiments.
\newblock {\em American psychologist}, 24(4):409.

\bibitem[Cattaneo et~al., 2017]{cattaneo2017}
Cattaneo, M.~D., Titiunik, R., and Vazquez-Bare, G. (2017).
\newblock Comparing inference approaches for {RD} designs: A reexamination of
  the effect of head start on child mortality.
\newblock {\em Journal of Policy Analysis and Management}, 36(3):643--681.

\bibitem[Chernoff, 1953]{chernoff1953}
Chernoff, H. (1953).
\newblock Locally optimal designs for estimating parameters.
\newblock {\em The Annals of Mathematical Statistics}, pages 586--602.

\bibitem[Clyde and Chaloner, 1996]{clydechaloner1996}
Clyde, M. and Chaloner, K. (1996).
\newblock The equivalence of constrained and weighted designs in multiple
  objective design problems.
\newblock {\em Journal of the American Statistical Association},
  91(435):1236--1244.

\bibitem[Cook and Wong, 1994]{cookwong1994}
Cook, R.~D. and Wong, W.~K. (1994).
\newblock On the equivalence of constrained and compound optimal designs.
\newblock {\em Journal of the American Statistical Association},
  89(426):687--692.

\bibitem[Dantzig and Wald, 1951]{dantzigwald1951}
Dantzig, G.~B. and Wald, A. (1951).
\newblock On the fundamental lemma of {Neyman} and {Pearson}.
\newblock {\em The Annals of Mathematical Statistics}, 22(1):87--93.

\bibitem[{De Oliveira}, 2018]{Oliveira2018}
{De Oliveira}, O. (2018).
\newblock The implicit function theorem for maps that are only differentiable:
  An elementary proof.
\newblock {\em Real Analysis Exchange}, 43(2):429--444.

\bibitem[Efron, 1971]{efron1971}
Efron, B. (1971).
\newblock Forcing a sequential experiment to be balanced.
\newblock {\em Biometrika}, 58(3):403--417.

\bibitem[Gelman and Imbens, 2017]{gelman2017high}
Gelman, A. and Imbens, G. (2017).
\newblock Why high-order polynomials should not be used in regression
  discontinuity designs.
\newblock {\em Journal of Business \& Economic Statistics}, 37(3):447--456.

\bibitem[Goldberger, 1972]{gold:1972}
Goldberger, A.~S. (1972).
\newblock Selection bias in evaluating treatment effects: Some formal
  illustrations.
\newblock Technical Report Discussion paper 128--72, Institute for Research on
  Poverty, University of Wisconsin--Madison.

\bibitem[Goldenshluger and Zeevi, 2013]{goldenshlugerzeevi2013}
Goldenshluger, A. and Zeevi, A. (2013).
\newblock A linear response bandit problem.
\newblock {\em Stochastic Systems}, 3(1):230--261.

\bibitem[Hu and Rosenberger, 2006]{hurosenberger2006}
Hu, F. and Rosenberger, W.~F. (2006).
\newblock {\em The theory of response-adaptive randomization in clinical
  trials}.
\newblock John Wiley \& Sons.

\bibitem[Hu et~al., 2015]{huetal2015}
Hu, J., Zhu, H., and Hu, F. (2015).
\newblock A unified family of covariate-adjusted response-adaptive designs
  based on efficiency and ethics.
\newblock {\em Journal of the American Statistical Association},
  110(509):357--367.

\bibitem[Jacob et~al., 2012]{jacob2012practical}
Jacob, R., Zhu, P., Somers, M.-A., and Bloom, H. (2012).
\newblock A practical guide to regression discontinuity.
\newblock {\em MDRC}.

\bibitem[Kluger and Owen, 2021]{klug:owen:2022:tr}
Kluger, D. and Owen, A.~B. (2021).
\newblock Tie-breaker designs provide more efficient kernel estimates than
  regression discontinuity designs.
\newblock Technical Report arXiv:2101.09605, Stanford University.

\bibitem[L{\"a}uter, 1974]{lauter1974}
L{\"a}uter, E. (1974).
\newblock Experimental design in a class of models.
\newblock {\em Mathematische Operationsforschung und Statistik},
  5(4-5):379--398.

\bibitem[L{\"a}uter, 1976]{lauter1976}
L{\"a}uter, E. (1976).
\newblock Optimal multipurpose designs for regression models.
\newblock {\em Mathematische Operationsforschung und Statistik}, 7(1):51--68.

\bibitem[Lee, 1987]{lee1987}
Lee, C. M.-S. (1987).
\newblock Constrained optimal designs for regressiom models.
\newblock {\em Communications in Statistics-Theory and Methods},
  16(3):765--783.

\bibitem[Lee, 1988]{lee1988}
Lee, C. M.-S. (1988).
\newblock Constrained optimal designs.
\newblock {\em Journal of Statistical Planning and Inference}, 18(3):377--389.

\bibitem[Lehmann and Romano, 2005]{lehmannromano2005}
Lehmann, E.~L. and Romano, J.~P. (2005).
\newblock {\em Testing statistical hypotheses}, volume~3.
\newblock Springer, New York.

\bibitem[Lipsey et~al., 1981]{lips:cord:berg:1981}
Lipsey, M.~W., Cordray, D.~S., and Berger, D.~E. (1981).
\newblock Evaluation of a juvenile diversion program: Using multiple lines of
  evidence.
\newblock {\em Evaluation Review}, 5(3):283--306.

\bibitem[Ludwig and Miller, 2007]{ludwigmiller2007}
Ludwig, J. and Miller, D.~L. (2007).
\newblock Does {Head Start} improve children's life chances? evidence from a
  regression discontinuity design.
\newblock {\em The Quarterly journal of economics}, 122(1):159--208.

\bibitem[Metelkina and Pronzato, 2017]{metelkinapronzato2017}
Metelkina, A. and Pronzato, L. (2017).
\newblock Information-regret compromise in covariate-adaptive treatment
  allocation.
\newblock {\em The Annals of Statistics}, 45(5):2046--2073.

\bibitem[Morrison and Owen, 2022]{mvtiebreaker}
Morrison, T.~P. and Owen, A.~B. (2022).
\newblock Optimality in multivariate tie-breaker designs.
\newblock Technical report, Stanford University.
\newblock arxiv2202.10030.

\bibitem[Neyman and Pearson, 1933]{neymanpearson1933}
Neyman, J. and Pearson, E.~S. (1933).
\newblock {IX}. {On} the problem of the most efficient tests of statistical
  hypotheses.
\newblock {\em Philosophical Transactions of the Royal Society of London.
  Series A}, 231(694-706):289--337.

\bibitem[Owen and Varian, 2020]{owen:vari:2020}
Owen, A.~B. and Varian, H. (2020).
\newblock Optimizing the tie-breaker regression discontinuity design.
\newblock {\em Electronic Journal of Statistics}, 14(2):4004--4027.

\bibitem[{R Core Team}, 2022]{R}
{R Core Team} (2022).
\newblock {\em R: A Language and Environment for Statistical Computing}.
\newblock R Foundation for Statistical Computing, Vienna, Austria.

\bibitem[Rosenberger and Sverdlov, 2008]{rosenbergersverdlov2008}
Rosenberger, W.~F. and Sverdlov, O. (2008).
\newblock Handling covariates in the design of clinical trials.
\newblock {\em Statistical Science}, 23(3):404--419.

\bibitem[Stigler, 1971]{stigler1971}
Stigler, S.~M. (1971).
\newblock Optimal experimental design for polynomial regression.
\newblock {\em Journal of the American Statistical Association},
  66(334):311--318.

\bibitem[Sverdlov et~al., 2013]{sverdlovetal2013}
Sverdlov, O., Rosenberger, W.~F., and Ryeznik, Y. (2013).
\newblock Utility of covariate-adjusted response-adaptive randomization in
  survival trials.
\newblock {\em Statistics in Biopharmaceutical Research}, 5(1):38--53.

\bibitem[Thistlethwaite and Campbell, 1960]{this:camp:1960}
Thistlethwaite, D.~L. and Campbell, D.~T. (1960).
\newblock Regression-discontinuity analysis: An alternative to the ex post
  facto experiment.
\newblock {\em Journal of Educational psychology}, 51(6):309.

\bibitem[Trochim and Cappelleri, 1992]{Trochim92}
Trochim, W.~M. and Cappelleri, J.~C. (1992).
\newblock Cutoff assignment strategies for enhancing randomized clinical
  trials.
\newblock {\em Controlled Clinical Trials}, 13(3):190--212.

\bibitem[Whittle, 1973]{whittle1973}
Whittle, P. (1973).
\newblock Some general points in the theory of optimal experimental design.
\newblock {\em Journal of the Royal Statistical Society: Series B
  (Methodological)}, 35(1):123--130.

\bibitem[Zhang et~al., 2007]{zhangetal2007}
Zhang, L.-X., Hu, F., Cheung, S.~H., and Chan, W.~S. (2007).
\newblock Asymptotic properties of covariate-adjusted response-adaptive
  designs.
\newblock {\em The Annals of Statistics}, 35(3):1166--1182.

\bibitem[Zhang and Hu, 2009]{zhanghu2009}
Zhang, L.-X. and Hu, F.-f. (2009).
\newblock A new family of covariate-adjusted response adaptive designs and
  their properties.
\newblock {\em Applied Mathematics-A Journal of Chinese Universities},
  24(1):1--13.

\end{thebibliography}

\end{document}